\documentclass{article}
\usepackage{graphicx} 
\usepackage[all,cmtip]{xy}
\usepackage{xcolor}

\usepackage{subcaption}
\usepackage[T2A]{fontenc}
\usepackage[export]{adjustbox}
\usepackage[utf8]{inputenc}
\usepackage[english]{babel}
\usepackage{amsmath,amssymb,amsfonts,amsthm,pifont,mathtools,mathrsfs}
\usepackage{amsmath}
\usepackage{cite,enumerate,float,indentfirst}
\graphicspath{{images/}{Figures/}}
\usepackage{tikz-cd} 
\usepackage{pgf}

\renewcommand{\Re}{\mathrm{Re}}
\renewcommand{\Im}{\mathrm{Im}}

\usepackage{hyperref}
\newcommand\restr[2]{{
  \left.\kern-\nulldelimiterspace 
  #1 
  \vphantom{\big|} 
  \right|_{#2} 
  }}

\newcommand{\RR}{\mathbb{R}}
\newcommand{\CC}{\mathbb{C}}

\newcommand{\ZZ}{\mathbb{Z}}

\newcommand{\pp}{\partial}

\newcommand{\bb}{\begin{equation}}
\newcommand{\ee}{\end{equation}}
\usepackage[english]{babel}

\usepackage{comment}

\DeclarePairedDelimiter\floor{\lfloor}{\rfloor}

\newtheorem{theorem}{Theorem}

\newtheorem{proposition}{Proposition}

\newtheorem{remark}{Remark}

\title{Limit shapes and harmonic tricks}
\author{Nikolai Kuchumov \thanks{%
    {\small Åbo Akademi University,
      \texttt{nikolai.kuchumov@abo.fi}}}}
\begin{document}

\maketitle
\begin{abstract}
This article has two main goals. First, it provides a self-contained exposition of the tangent plane method for the dimer model — a technique for analyzing arctic curves and limit shapes introduced by R. Kenyon and I. Prause (2020). Second, it extends this method to multiply connected domains through a nontrivial computation of the frozen boundary for the Aztec diamond with a hole. This computation yields the first explicit parametrization in terms of elliptic functions of a family of arctic curves of a multiply-connected region indexed by the height change (hole height). We also derive and visualize the corresponding limit height functions.
\end{abstract}

\section{Introduction}
\label{sect:1}
The planar dimer model has been an active area of research for several decades. It originated in the 1960s in the physics literature on equilibrium statistical mechanics, in the works of Kasteleyn and Temperley–Fisher \cite{Kasteleyn,TF}. The model later attracted the attention of the mathematical community through studies of the Aztec diamond and its connection to alternating sign matrices \cite{Propp_aztec}. This naturally led to the question: What does a typical domino tiling of a large planar domain look like?
The answer for the large Aztec diamond was given by the Arctic Circle Theorem \cite{JPS}. Subsequent work extended this result to other simply connected regions \cite{CKP} and to various integrable domains generalizing the Aztec diamond \cite{BK,Colomo-Sportiello,Di_francesco_guitter_aztec}. The first analysis of a multiply connected region appeared in \cite{BG}, and the corresponding variational principle was extended in our previous work \cite{Kuchumov:dominoes}. In the present paper, we perform the first explicit computation of a one-parameter family of limit shapes, together with the arctic curves of a multiply connected version of the Aztec diamond parametrized by a modular parameter $\tau$. This region is a large Aztec diamond order $N$ with a hole in the center consisting of a smaller Aztec diamond of order $\floor{\varkappa(\tau) N}$. For each $\tau$ we compute the asyptotics of a typical configuration of domino tilings of domain with the hole size $\varkappa(\tau)$. Equivalently, one could numerically reverse the dependence and view $\varkappa$ as the parameter of the family as it is the only "real" physical parameter of the domain.

\begin{figure}
    \centering
    \includegraphics[width=0.5\linewidth]{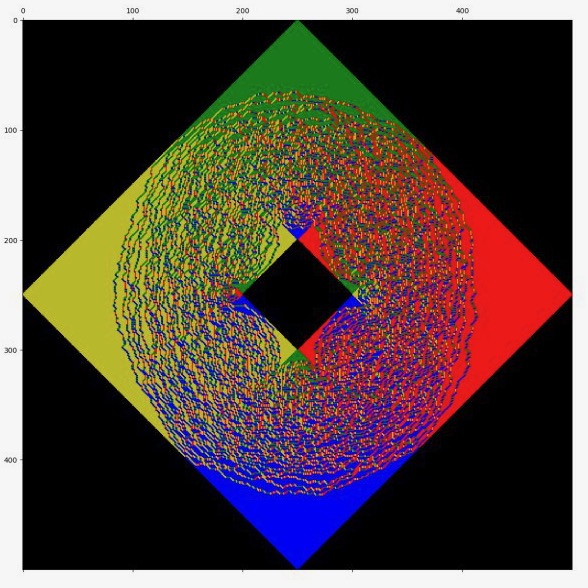}
    \caption{Computer simulation of a random domino tiling of Aztec diamond of order $AD_{500}$ with a centered hole consising of $AD_{125}$.}
    \label{fig:aztec_holey_sim}
\end{figure}

The dimer model deals with a finite bipartite graph \(\Gamma\), and a \emph{dimer configuration} (or \emph{dimer cover}) on \(\Gamma\) is a subset of its edges such that each vertex is incident to exactly one edge in the subset. When \(\Gamma\) is a subgraph of the square grid, dimer covers are in bijection with domino tilings of the dual graph of \(\Gamma\). Abusing notation, we will work with this graph instead. We think of it as a union of unit squares of the square grid. Let us also fix a chessboard coloring, so that we have black and white squares (they coexist).

Then, let us fix a uniform distribution on the set of domino tilings of $\Gamma$.
Turns out that the typical domino tiling of a large region exhibit a phenomenon similar to the Aztec diamond.
The Aztec Circle Theorem states that a random domino tiling of a large Aztec diamond is fixed outside the inscribed unit circle to the rescaled domain, while it remains random inside it.
The inner region is the \emph{rough region} (liquid region), the asymptotically deterministic regions are \emph{frozen regions}, and the curve separating them is the \emph{arctic curve} (frozen curve), which is a unit circle for the Aztec diamond, see \ref{fig:uni_aztec} for a computer simulation. In other regions the same phenomenon(separation of domain into rough and frozen regions) typically occurs.

\begin{figure}
    \centering
    \includegraphics[width=0.5\linewidth]{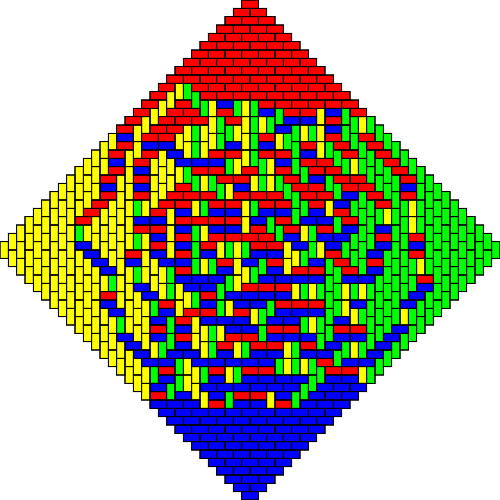}
    \caption{Uniformly random domino tiling of an Aztec diamond of order 30.}
    \label{fig:uni_aztec}
\end{figure}

A common tool to address this question is to encode a domino tiling by the \emph{height function}. It is a function $H_D:V(\Gamma)\to \ZZ$ defined by a certain combinatorial rule.
It turns out that in the limit of a large region, the re-scaled height function approximates a continuous Lipschitz function $\mathfrak{h}$ with overwhelming probability.
Moreover, in the frozen regions this function is linear with the maximal allowed slope, while in the rough region its gradient lies strictly inside the Newton polygon.
The main result of \cite{CKP}, the variational principle for random domino tilings, characterize $\mathfrak{h}$ as a minimizer of a certain variational function of the form $\int_{\Omega}\sigma(\nabla \mathfrak{h})dxdy$. Here $\sigma:\mathcal{N}\to\RR$ is a convex function called \emph{surface tension}, and $\mathcal{N}$ is an integer polygon, the so-called \textit{Newton polygon} of the model \cite{KOS}, which is the set of allowed slopes $\nabla \mathfrak{h}\in\mathcal{N}$.


The generality of the variational approach is offset by the difficulty of its implementation: the singular behavior of 
$\sigma$ complicates the associated Euler–Lagrange equation. The next step of understanding of this this equation was done in \cite{KO:burgers}, where the authors lowered the order of the Euler–Lagrange equation for lozenge tiling model by introducing the so-called \textit{complex gradient}(complex slope) map $z$. It satisfies the complex Burgers equation and the spectral curve condition in the liquid region, where it also defines a non-trivial complex structure. When it is approaching $\pp\mathcal{L}$, it gets a singularity. In the interior of $\mathcal{L}$, expressed in complex-gradient coordinates, the Kenyon--Okounkov conjecture predicts that the fluctuations are governed by the the pullback of the Gaussian free field from the upper-half plane $\mathbb{H}$ \cite{KO:burgers}. The authors also give an argument of existance of an algebraic arctic curve for a family of regions in the lozenge tiling model using methods of tropic geometry. The mathematically regorous work on existance and algebraic arctic curve for the dimer model was done in \cite{ADPZ}.

On the other hand, researchers have found ways to compute the limit shape using various methods, such as the Schur process introduced in \cite{Okounkov_Reshetikhin} and further developed through representation-theoretic approaches in \cite{BF,BG}. More recent developments include the tangent method, which also applies for the six-vertex model \cite{Colomo-Sportiello,D2}.

Then, the work \cite{KP:plane} focuses on a systematic approach to the study of limit shapes from the variational perspective. 
It realizes the graph of $\mathfrak{h}^{\star}$ as the envelope of its tangent planes
\[
\mathcal{P}_{x_0,y_0}
=\{(x,y,z)\in\mathbb{R}^3 \mid s(x_0,y_0)x+t(x_0,y_0)y+c(x_0,y_0)=z\},
\]
supplemented by an appropriate tangency condition.

The authors parametrize points $(x,y)\in\mathcal{L}$ by a single complex variable, the so-called intrinsic coordinate $u=u(x,y)$, which is a ramified cover of the complex slope $z$. This intrinsic coordinate uniformizes the liquid region $\Sigma \to \mathcal{L}$.
Then, by Theorem 1 in \cite{KO:burgers} for random lozenge tilings and Theorem 4.8 \cite{ADPZ} for the generic periodic dimer model, the gradient coordinates $s$ and $t$ are harmonic as functions of $u$. Moreover, the intercept function $c$ is also harmonic by Theorem 3.1 \cite{KP:plane}.

Then, the frozen boundary $\partial\mathcal{L}$ can be obtained parametrically from $\partial\Sigma$ via the tangency condition. One first constructs the harmonic extensions of $s(u)$, $t(u)$, and $c(u)$. Then, for each $u\in\partial\Sigma$, one solves the complex tangency equation.
\[
s_u(u)\,x + t_u(u)\,y + c_u(u) = 0
\]
for the real $(x,y)$. As this is a single complex equation, it imposes two real constraints and determines $x(u)$ and $y(u)$. The solution $(x(u),y(u))$ parametrizes the frozen boundary for $u\in\pp\Sigma$. For internal points $u\in \Sigma\setminus\pp\Sigma$ we obtain $(x(u),y(u))\in\mathcal{L}(\Omega)$ in the liquid region, and the height function as a function of $u$, $\mathfrak{h}(u)=s(u)x(u)+t(u)y(u)+c(u)$, the limit shape can be computed and plotted numerically. Such a function $\mathfrak{h}$ is therefore the natural candidate for the solution to the variational problem. To make this statement rigorous, however, one must also verify that $\mathfrak{h}$ minimizes the functional on the frozen regions. This can be done, for example, by adapting the proof of Theorem 9.3 in \cite{ADPZ}; see Section 9 there for details. Alternatively, one may use Proposition 23 of \cite{Bobenko_dimers}.

\begin{figure}[ht]
    \centering
    \begin{minipage}[t]{0.48\textwidth}
        \centering
        \includegraphics[width=\linewidth]{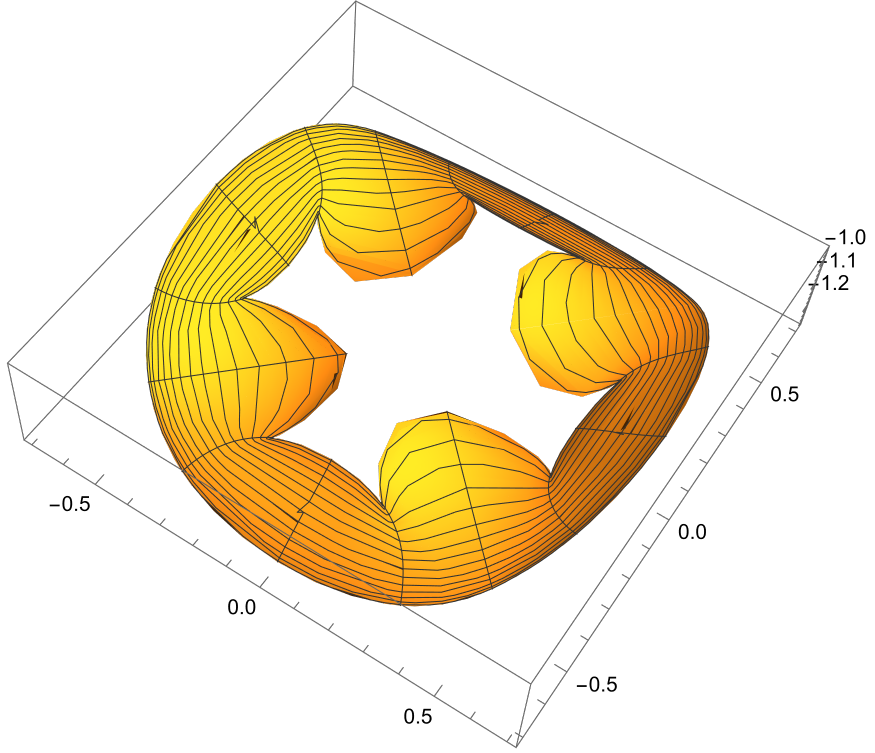}
        \caption{Plot of $\mathfrak{h}$ for the Aztec diamond with a hole}
        \label{fig:placeholder}
    \end{minipage}
    \hfill
    \begin{minipage}[t]{0.48\textwidth}
        \centering
        \includegraphics[width=\linewidth]{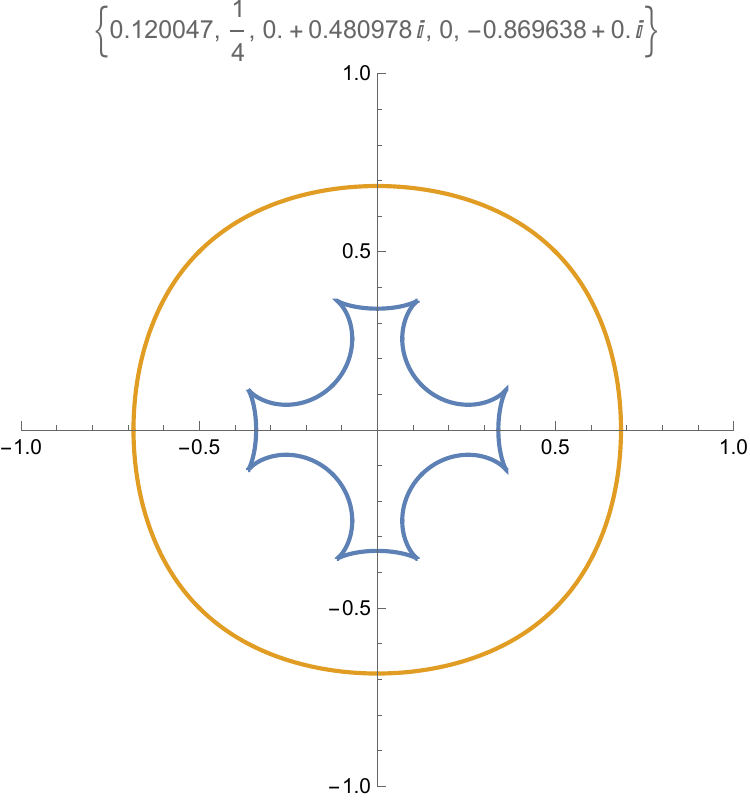}
        \caption{Arctic curve of the Aztec diamond with a hole. The yellow curve is the outer connected component, while the blue is the inner one. The corresponding values of parameters $a,\varkappa,\tau,\delta$ are written above the plot.}
        \label{fig:aztec_holey}
    \end{minipage}
\end{figure}

In the work \cite{KP:plane} the authors demonstrate harmonicity from the so-called trivial potential property of the dimer model. We do it in a slightly different way, and in Proposition~\ref{harmonicity} and Theorem~\ref{theorem:curve} we show it directly from the Harnack property of the associated spectral curve. Thus, our proof is valid as long as the assosiated spectral curve is a Harnack curve, for instance for periodic weights. It provides a clear connection with the work \cite{Borodin-Berggren} later extended to generic spectral curves in \cite{Bobenko_dimers}. The connection is that while we operate with zeros of a differential  $dF=xds+ydt+dc$, in \cite{Borodin-Berggren} the roles of $s$ and $t$ are exchanged compared with our notations. In work \cite{Bobenko_dimers} the corresponding objects are $xds=d\zeta_1$,$ydt=d\zeta_2$ and $dc=d\zeta_3$. These two articles deal with so-called gaze bubbles, which make a liquid region multiply-connected. However, the underlying domain is still simply-connected unlike our situation.

In the present article, we extend the tangent plane method to a multiply connected domain and compute the family of arctic curves for the \textit{Aztec diamond with a hole}. While the variational principle for this case was extended in our previous result \cite{Kuchumov:dominoes}, here we provide a concrete example of a family of multiply-connected limit shapes.
The Tangent plane method gives a solution to the variational principle, and it operates directly with continuous objects. Therefore, we do not need to compute discrete data and then pass to the limit that quite often requires the most technical effort, for instance, in the tangent method \cite{Colomo-Sportiello,D2}. The tradeoff of this advantege is that we need to match the critical points of $s_u(u), t_u(u)$ and $c_u$ inside $\mathcal{L}$, which takes a big effort. One could produce plots of arctic curves without it, however, these curves would be "fake" as they would not correspond to any $3d$ surface minimizing the functional. Note that the difficulty of critical points emerges only for multiply-connected $\Omega$, and we do not see it in gas bubbles(there are no such critical points), where $\mathcal{L}$ is also multiply-connected.

We show that in the limit, both the limit shape $\mathfrak{h}$, and $\varkappa$ are expressed in terms of elliptic functions. In particular, we provide a numeric scheme of the computation in  which works also for the limit shape of lozenge tilings of a hexagon with a hexagonal hole, which will be in the next
For instance, $\varkappa$ is given by the value of an elliptic function evaluated at a particular point, see \eqref{elliptic_size}.
Furthermore, as $\tau\to\infty$ harmonic extensions $s,t,c$ converge to the ones of usual Aztec diamond with cylindrical parametrization due to the asymptotic formula for $\sigma$ function, see table VII \cite{akhiezer}, or by a brownian motion argument.

\subsection{Structure of the paper}
The structure of the paper is the following,
\begin{itemize}
    \item In Section \ref{sect:2} we give a review of the complex gradient map in Theorem \ref{theorem:curve} and then proceed to the conformal coordinates in Proposition \ref{prop2.1}. Then, we show the harmonicity of the functions $s$, $t$, and $c$ in Proposition \ref{harmonicity}. We conclude with the tangent equation \ref{tangent_eq} which will be used to reconstract the arctic curve.
    \item In Section \ref{sect:3} we discuss applications to the domino tilings. In subsection \ref{sect_aztec} we discuss limit shape of the Aztec diamond, and then in subsection \ref{doubly_aztec} we discuss Aztec diamond with doubly-periodic weights with nessasary background on elliptic functions in \ref{elliptic_appendix}. After it, we discuss the main example, Aztec diamond with a hole in subsection \ref{Our_case}.
    \item Appendix \ref{appendix:beltrami} is devoted to derivation of the Beltrami equation in the needed forms.
\end{itemize}
\subsection{Acknowledgements}
We would like to thank Cédric Boutilier for his help throughout this work. This work was supported by the Research Council of Finland (grant 365297). Moreover, we thank Rick Kenyon and István Prause for providing Mathematica code for the doubly periodic Aztec diamond, which was crucial for the work. Furthermore, we thank Tomas Berggren and Mattias Eriksson for the discussions that led to our numerical scheme for Aztec diamond with a hole.


\section{Assumptions and geometry of the tangent plane method}
\label{sect:2}
We assume the number of frozen regions around each connected boundary component, and the topological type of the liquid region. Then, we write a table consisting of columns with boundary planes at each frozen region. After that we perform harmonic extensions for $s,t$ and $c$ depending on parameters to be fixed later. Then, we analize the critical points of derivatives $s_u,t_u $ and $c_u$, and in order to match them we fix the parameters, this is the difficult part, which is partially availible only numerically. After that, we are able to produce the figure of the Arctic curve, and the graph of $\mathfrak{h}$.

Also, we are using the Newton polygon for the square lattice turned by 45 degrees so that the characteristic
polynomial is $P(z,w)=1+z+w-zw$ and the Newton polygon is the unit square with vertices $(0,0),(1,0),(0,1), (1,1)$.

We have the following picture going from the boundary of the region to the bulk: near the boundary $\pp\Omega$, we have frozen, where $\mathfrak{h}$ is linear with the maximal possible slope, $\mathfrak{h}\in\pp\mathcal{N}$. The boundary of this part is the arctic curve $\mathcal{C}$,
which is smooth except at finitely many points.
After crossing the boundary, we end up in the liquid region, where $\nabla
\mathfrak{h}\in\mathring{\mathcal{N}}$. There, we might have another connected
component of the arctic curve $\mathcal{C}$, which bounds a gas phase where the gradient is constant $\nabla\mathfrak{h}=q$.
Now, we know enough properties of $\pp \mathcal{L}$
to perform some computations, and the main tool are conformal coordinates, which help to simplify the analysis of the Euler-Lagrange equations for the dimer model.

The idea is that we start with the liquid region $\mathcal{L}$, and then
assuming that $\mathfrak{h}$ is the minimizer, we map the liquid region to the
Newton polygon, $(x,y)\mapsto (s,t):=\nabla \mathfrak{h}(x,y) \in \mathcal{N}$.
Moreover, in the paper~\cite{KP:plane}, the authors look at $\mathcal{N}$ as a
Riemann surface with the metric $g$ given by the Hessian of surface tension
$\sigma$:
\bb
g=\sigma_{ss} ds^2+2\sigma_{st}ds dt+\sigma_{tt}dt^2.
\ee
Convexity of $\sigma$ implies that the Hessian is positive definite,
and thus corresponds to a (non-degenerate) metric.

Then, they perform the main trick of the method, that is a specific change of coordinates, $(s,t)\mapsto z$,
for complex variables $z$ on the spectral curve $P(z,w)=0$.
The relation between $(s,t)$ and $(z,w)$ can be formulated graphically as on Figure~\ref{fig:rectangle_zw}. We justify this figure in Theorem~\ref{theorem:curve} below.

\begin{figure}
    \centering
    \def\svgwidth{4cm}
\begingroup%
  \makeatletter%
  \providecommand\color[2][]{%
    \errmessage{(Inkscape) Color is used for the text in Inkscape, but the package 'color.sty' is not loaded}%
    \renewcommand\color[2][]{}%
  }%
  \providecommand\transparent[1]{%
    \errmessage{(Inkscape) Transparency is used (non-zero) for the text in Inkscape, but the package 'transparent.sty' is not loaded}%
    \renewcommand\transparent[1]{}%
  }%
  \providecommand\rotatebox[2]{#2}%
  \newcommand*\fsize{\dimexpr\f@size pt\relax}%
  \newcommand*\lineheight[1]{\fontsize{\fsize}{#1\fsize}\selectfont}%
  \ifx\svgwidth\undefined%
    \setlength{\unitlength}{205.79583356bp}%
    \ifx\svgscale\undefined%
      \relax%
    \else%
      \setlength{\unitlength}{\unitlength * \real{\svgscale}}%
    \fi%
  \else%
    \setlength{\unitlength}{\svgwidth}%
  \fi%
  \global\let\svgwidth\undefined%
  \global\let\svgscale\undefined%
  \makeatother%
  \begin{picture}(1,0.99999989)%
    \lineheight{1}%
    \setlength\tabcolsep{0pt}%
    \put(0,0){\includegraphics[width=\unitlength,page=1]{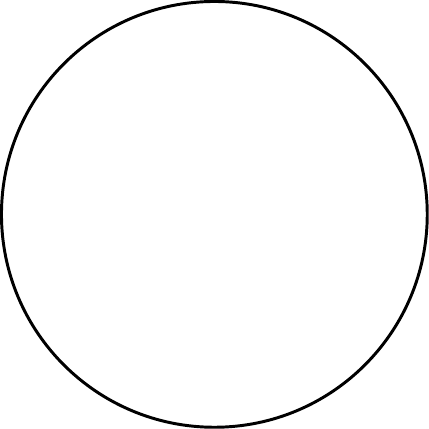}}%
    \put(0.42299461,0.82095331){\color[rgb]{0,0,0}\makebox(0,0)[lt]{\lineheight{1.25}\smash{\begin{tabular}[t]{l}$zw$\end{tabular}}}}%
    \put(0.80878257,0.44563278){\color[rgb]{0,0,0}\makebox(0,0)[lt]{\lineheight{1.25}\smash{\begin{tabular}[t]{l}$z$\end{tabular}}}}%
    \put(0.16323098,0.48820459){\color[rgb]{0,0,0}\makebox(0,0)[lt]{\lineheight{1.25}\smash{\begin{tabular}[t]{l}$w$\end{tabular}}}}%
    \put(0.4301903,0.04970197){\color[rgb]{0,0,0}\makebox(0,0)[lt]{\lineheight{1.25}\smash{\begin{tabular}[t]{l}$1$\end{tabular}}}}%
    \put(0,0){\includegraphics[width=\unitlength,page=2]{angle_new321.pdf}}%
    \put(0.79480155,0.30895818){\color[rgb]{0,0,0}\makebox(0,0)[lt]{\lineheight{1.25}\smash{\begin{tabular}[t]{l}$\pi t$\end{tabular}}}}%
    \put(0.14319493,0.29683768){\color[rgb]{0,0,0}\makebox(0,0)[lt]{\lineheight{1.25}\smash{\begin{tabular}[t]{l}$\pi s$\end{tabular}}}}%
  \end{picture}%
\endgroup%

    \caption{Relation between coordinates $(s,t)$ and $(z,w)$}
    \label{fig:rectangle_zw}
\end{figure}

In a simple example like Aztec diamond, one could use such $z$ as a conformal coordinate and parametrize points of $\pp\mathcal{L}$ by it. In other words, we can encode points of $\in\mathcal{L}$ by the positive half of the spectral curve $\mathcal{C}^{+}\simeq \mathbb{H}$.
However, in a more compicated situatuion, where degree of $\nabla \mathfrak{h}$ is higher than $1$, for instance in a multiply-connected domain, one would need to use a ramified cover of $z$, $\Sigma \to \mathcal{C}$. Then, the conformal model of $\mathcal{L}$ is not $\mathbb{H}$, but an annulus $\Sigma^{+}\simeq[0,2]\times[0,\tau]/\sim $ where we identify $u$ and $u+2$. There $z$ becomes a meromorphic function of $u\in\Sigma$, and $s,t$ with $c$ are harmonic maps of this $u$ (as the first two are harmonic in terms of $z$, and harmonicity of $c$ follows from the Harnack property of the spectral curve).
The following commutative diagram illustrates this:

\[
\begin{tikzcd}
  {\mathcal{L}} \arrow{d}{1:1}\arrow{r}{\nabla \mathfrak{h}} & \mathcal{N} \arrow{d}{1:1} \arrow{r}{\nabla \sigma}& \mathcal{A}\\
  \Sigma^+ \arrow[d, phantom, sloped, "\subset"] \arrow{r}{\text{deg}\, d} & \mathcal{C}^{+} \arrow{ru}[swap, pos=0.4]{(\log |z|,\log |w|)} \arrow[d, phantom, sloped, "\subset"] \\
  \Sigma \arrow{r}{\pi} & \mathcal{C}
\end{tikzcd}
\]

There is also a creterion for an instrinsic coordinate from \cite{KP:plane}, that we show for completness in Proposition \ref{prop2.1}.

We begin our proof of the diagram by showing Theorem \ref{theorem:curve}, which explains the relation between $\mathcal{C}^{+}$ and $\mathcal{L}$.

\begin{theorem}
    In the liquid region $\mathcal{L}$,
    there are defined two functions $z(x,y)$, $w(x,y)$ satisfying the following
    properties,
    \begin{enumerate}
        \item $\nabla \mathfrak{h}=\frac{1}{\pi}(\arg w, -\arg z )$,
        \item Ampère's equation $\frac{z_x}{z}+\frac{w_y}{w}=0$,
        \item spectral curve equation $P(z,w)=0$.
    \end{enumerate}
    \label{theorem:curve}
\end{theorem}

\begin{proof}
We start with a point $(z_0,w_0)$ on the spectral curve $\{P(z,w)=0\}$. Then, we know
that the real part of its image under the logarithm map $(\log |z_0|,\log |w_0|)$
defines a point in the amoeba $(X,Y):=(\log |z_0|,\log |w_0|)\in\mathcal{A}$, see also \cite{KOS}[3.2.3].
Here, we assume that $(z_0,w_0)$ is a generic point, so that $(X,Y)=(\log|z_0|,\log
|w_0|)$ is in the interior of $\mathcal{A}$. We further assume that
$z_0\in\mathbb{H}$ (so that $w_0$ has a negative imaginary part).~In fact, we know from the Harnack
property of the spectral curve that there is exactly another root of $P$, which
is sent to the same point of the amoeba, given by $(\bar{z_0},\bar{w_0})$.
We write
\begin{equation*}
  z_0 = e^{X+i\theta_0},\qquad w_0 = e^{Y+i\omega_0},
\end{equation*}
with $\theta_0\in(0,\pi)$ and $\omega_0\in(\pi,2\pi)$.

The gradient of the Ronkin function $\mathcal{R}$ evaluated at this point defines the corresponding slope $(s,t)$.
More precisely, we have the following equation,
\bb
(s,t):=\nabla \mathcal{R}(\log|z|,\log |w|).\label{eq:slope}
\ee

Let us compute the derivatives of the $\mathcal{R}$ by definition,
\begin{multline}
s=\frac{\partial}{\partial X} \frac{1}{(2\pi i)^2}\iint_{\mathbb{T}^2} \log P(e^X z, e^Yw)\frac{dzdw}{zw}=\\
\frac{1}{(2\pi i)^2}\int_{|w|=e^Y} \left( \int_{|z|=e^X} \frac{\partial_XP(e^Xz,e^Yw)}{P(e^X z, e^Yw)} \frac{dzdw}{zw} \right)\\
=\frac{1}{2\pi i} \int_{|w|=e^Y}\frac{dw}{w}\left( \frac{1}{2\pi
i}\int_{|z|=e^X}d \log P \right),
\end{multline}
where we look at $\left( \frac{1}{2\pi i}\int_{|z|=e^X}d\log P \right)$ as a function of $w$, which is fixed, and then integrate it over the circle $\{|w|=e^Y\}$.
For a given value of $w$, there is a unique value of $z=z(w)=\frac{w+1}{w-1}$
such that $P(z(w),w)=0$. It turns out that when $w$ moves counterclockwise
around the circle of radius $e^Y$, the root $z(w)$ is inside the disc of radius
$e^X$ when $w$ is on the arc from $\bar{w_0}$ to $w_0$, and outside when $w$ is
on the arc from $w_0$ to $\bar{w_0}$.
As a consequence, the function $w\mapsto\int_{|z|=e^X} d\log P(\cdot,w)$ is a
the indicator function
of the arc $(\bar{w_0}, w_0)$.


Therefore, the remaining integral is 

\begin{equation*}
  s=
\int_{\bar{w}_0}^{w_0} \frac{1}{2 \pi i}
\frac{dw}{w}=\frac{1}{2\pi i} (\log w_0 -\log \bar{w}_0)=\frac{1}{2i\pi}(\log (e^{Y+i
\omega_0})-\log ( e^{Y+i(2\pi- \omega_0)}))=\frac{\omega_0}{\pi}-1.
\end{equation*}

The computation for the other coordinate of the gradient corresponding to $t$ is
the same except that the integral
$\frac{1}{2i\pi}\int_{|w|=e^Y} d\log P(z,\cdot)$ is the indicator function of
the oriented arc from $z_0$ to $\bar{z_0}$. This means that when we integrate
the result over $z$, we get
\begin{equation*}
  t=
\int_{z_0}^{\bar{z_0}} \frac{1}{2 \pi i}
\frac{dz}{z}=\frac{1}{2\pi i} (\log \bar{z_0} -\log z_0)=\frac{1}{2i\pi}(\log (e^{X+i
(2\pi-\theta_0)})-\log ( e^{X+i\theta_0)}))=1-\frac{\theta_0}{\pi}.
\end{equation*}


\begin{figure}
    \centering
    \includegraphics[width=0.5\linewidth]{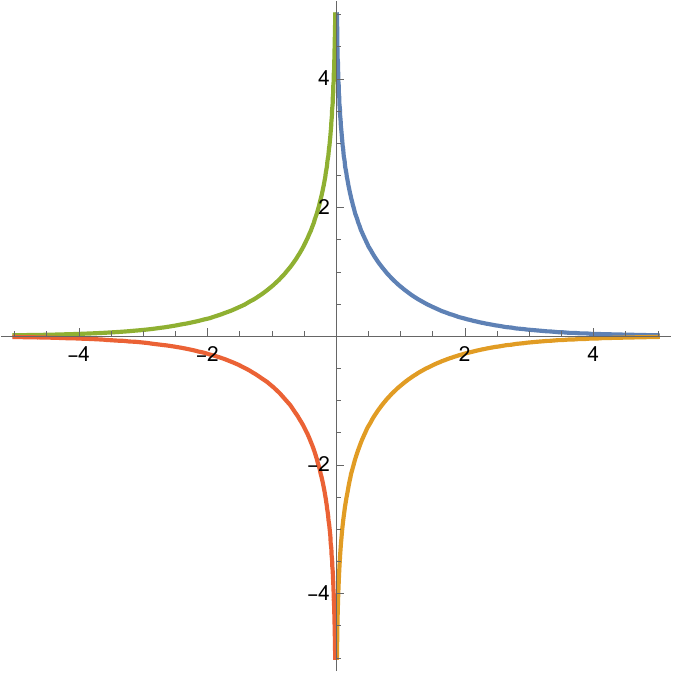}
    \caption{Plot of Amoeba for $\ZZ^2$ and our choice of $P(z, w)$. Quadrants correspond to the
frozen regions.}
    \label{fig:amoeba}
\end{figure}

As a result, we have that

\bb
\label{pre-beltrami0}
\log z_0=X+\pi i (1-t),\qquad
\log w_0= Y+\pi i (1+s).
\ee
By a global translation of the set of slopes by the vector $(1,-1)$ (or by
  replacing $z$ by $-z$ and $w$ by $-w$), we can
assume that the following relation holds
\begin{equation}
\label{pre-beltrami}
  \log z_0 = X-i\pi t,\qquad \log w_0 = Y+i\pi s.
\end{equation}

We now associate to each point $(x,y)$ of the liquid region a pair $(z,w)$, in
such a way that the arguments of $z(x,y)=z_0$ and $w(x,y)=w_0$ match the
(translated) slopes $-t$ and $s$ respectively, times $\pi$ according to
\eqref{pre-beltrami}.

Now, we substitute $X,Y$ from~\eqref{pre-beltrami} to the intrinsic Euler
Lagrange equation, which can be rewritten as $X_x+Y_y=0$
and see that it becomes $\Re
(\frac{z_x}{z}+\frac{w_y}{w})=0$. This is one real equation, which should be
supplemented with a consistency condition $\frac{\partial s}{\partial
y}=\frac{\partial t}{\partial x}=\frac{\partial^2 \mathfrak{ h}}{\partial y
\partial x}$ or $-\frac{\partial s}{\partial y}+\frac{\partial t}{\partial
x}=0$, which by~\eqref{eq:slope} can be written as $\langle(-\frac{\partial}{\partial
y},\frac{\partial}{\partial x}), \nabla \mathcal{R}(\log|z|,\log |w|)\rangle=0$, since
$\nabla \mathcal{R}(\log|z|,\log |w|)=\frac{1}{\pi}(+\arg w, -\arg z)$. Thus, we
end up with the second real equation $\Im (\frac{z_x}{z}+\frac{w_y}{w})=0$. It
completes the derivation of the Ampère's equation. 
\end{proof}

It turns out that the coordinate $z$ itself defines conformal coordinates $(U,V):=(\Re z, \Im z)$. 
In fact, let us look at~\eqref{pre-beltrami}. 
Since logarithm is a holomorphic map,
it satisfies the Cauchy-Riemann equations, and so does the right-hand side.
For example, for $Y+i\pi s$ we have two equations,
\bb
Y_U=\pi s_V,\qquad
Y_V=-\pi s_U.
\ee
and for $X-i\pi t$.

\bb
X_U=-\pi t_V,\qquad
X_V=\pi t_U.
\ee

These equations are combined into complex equations with the help of Wirtinger derivatives,

\begin{equation*}
X_z=\frac{1}{2}(X_U-iX_V)=\frac{1}{2}(-i\pi (s_U-i s_V))=-i\pi s_z,
\end{equation*}
\bb
Y_z=i\pi t_z. \label{Holomorphic_XY}
\ee
We also see that $\frac{Y_z}{s_z}=i\pi=-\frac{X_z}{t_z}$, therefore, we have $\frac{X_z}{t_z}+\frac{Y_z}{s_z}=0$. It is the first criterion for $z$ being an intrinsic variable from proposition 2.1 from~\cite{KP:new}.

Then, recall that $X=\sigma_s (s,t), Y=\sigma_t(s,t)$, which enables us writing
\bb
\pi s_V=Y_U=\sigma_{st}s_U+\sigma_{tt}t_U,
\pi t_V=-X_U=-\sigma_{ss}s_U-\sigma_{st}t_U.
\ee
This is nothing but the real Beltrami equation, a sufficient condition for the
coordinates $(U,V)$ to be conformal/conformal. Thus, $z=U+iV$ is the intrinsic
coordinate according to the definition of~\cite{KP:plane}. The authors there
propose a criterion for a coordinate $\zeta$ to be an intrinsic one, that we
formulate in the following Proposition:
\begin{proposition}[\cite{KP:plane}, Proposition 2.1]
    The following two equations are each equivalent to $\zeta$ being an intrinsic coordinate
    \bb
    \frac{X_{\zeta}}{s_\zeta}+\frac{Y_{\zeta}}{t_\zeta}=0 \label{equivalence}
    \ee
\bb
\frac{s_\zeta}{t_\zeta}=\frac{-\sigma_{st}-i\sqrt{\det H_\sigma}}{\sigma_{ss}}
\ee
And in the case the above equations hold, we have 
\bb
\frac{X_\zeta}{t_\zeta}=-i \det H_\sigma=-\frac{Y_\zeta}{s_\zeta}.
\ee
\label{prop2.1}
\end{proposition}
Note also that the last relation from \hyperref[prop2.1]{Proposition~\ref{prop2.1}} together with~\eqref{Holomorphic_XY} implies that $\det H_\sigma=\pi^2$.
\begin{proof}
Let us define $\gamma:=\frac{s_\zeta}{t_\zeta}$. Then, due to $X_\zeta=(\sigma_s)_\zeta=\sigma_{ss} s_\zeta+\sigma_{st} t_\zeta$, we write $\frac{X_\zeta}{t_\zeta}=\sigma_{ss}\gamma+\sigma_{st}$. Rewriting this computation for $Y$, we see
$\frac{Y_\zeta}{s_\zeta}=\sigma_{tt} \frac{1}{\gamma}+\sigma_{st}$. Therefore, the first condition of \hyperref[equivalence]{Proposition~\ref{equivalence}} is nothing but

\bb
\sigma_{ss}\gamma^2+2\sigma_{st}\gamma+\sigma_{tt}=0.
\ee
The two roots are

\bb
\gamma=\frac{-\sigma_{st}\pm i\sqrt{\det H_\sigma}{}\sigma_{ss}}{\sigma_{ss}}.
\ee

Since in our convention, $\zeta$ is assumed to be orientation-reversing map, $\Im(\gamma)<0$ and we have the minus sign.
Furthermore, $\frac{X_\zeta}{t_\zeta}=-i\sqrt{\det H_\sigma}=\frac{Y_\zeta}{s_\zeta}$.
\end{proof}

Now, with the help of \hyperref[prop2.1]{Proposition~\ref{prop2.1}} we can prove the intrinsic Euler-Lagrange equation in two forms,
\begin{proposition}
For $(x,y)\in\mathcal{L}$ we have
    \bb
    X_\zeta \zeta_x+Y_\zeta \zeta_y=0,
    \ee
    \bb
    \frac{\zeta_x}{\zeta_y}=\frac{s_\zeta}{t_\zeta}=\gamma.
    \ee
    \label{intrinsic}
\end{proposition}

\begin{proof}
    The last expression of the proof of \hyperref[prop2.1]{Proposition \ref{prop2.1}} implies that (we are using the chain rule for $X_x=X_\zeta \zeta_x+X_{\bar{\zeta}} \bar{\zeta}_x$ and $\bar{f}(\bar{z})=\bar{f(z)})$
    ,
    \bb
    X_x-i\sqrt{\det H_\sigma} t_x=X_\zeta \zeta_x+ X_{\bar{\zeta}}\bar{\zeta}_x-i\sqrt{\det H_\sigma}(t_\zeta \zeta_x+t_{\bar{\zeta}}\bar{\zeta}_x)=
    \ee
    \bb
    =(X_\zeta-i\sqrt{\det H_\sigma} t_\zeta)\zeta_x+ \overline{(X_\zeta+i\sqrt{\det H_\sigma}t_\zeta)\zeta_x}=2X_\zeta \zeta_x.
    \ee
    
    Similarly for $Y$,
\bb
Y_y+i\sqrt{\det H_\sigma}s_y=2Y_\zeta \zeta_y.
\ee
Therefore, since $t_x=s_y$ we have the following intrinsic Euler-Lagrange equation

\bb
X_\zeta \zeta_x+Y_\zeta \zeta_y=0.
\ee

Let us look at \eqref{pre-beltrami}, the left-hand side $X-i\pi t$ is a holomorphic function of $\zeta$, as well as $Y+i\pi s$. Therefore, they satisfy the Cauchy-Riemann equation, which tells us that 
\begin{align}
    X_\zeta &= i \pi t_\zeta,\\
    Y_\zeta &= -i \pi  s_\zeta.
\end{align}
Which in combination with  $X_\zeta \zeta_x+Y_\zeta \zeta_y=0$ gives us that $t_\zeta \zeta_x-s_\zeta \zeta_y=0$ or $\frac{t_\zeta}{s_\zeta}=\frac{\zeta_y}{\zeta_x}=\frac{1}{\gamma}$.

Then, apply the \hyperref[intrinsic]{Proposition \ref{intrinsic}}
 combined with the relation $-\frac{y_{\bar{z}}}{x_{\bar{z}}}=\frac{z_x}{z_y}$.
 It follows from $z=z(x,y)$, where we take the derivative with respect to
 $\bar{z}$ and obtain $0=z_x x_{\bar{z}}+z_y y_{\bar{z}}$. Finally, we get
 $-\frac{y_{\bar{z}}}{x_{\bar{z}}}=\frac{z_x}{z_y}=\frac{s_z}{t_z}$, and we are
 done. 
\end{proof}

However, one can see that $z$ is a conformal coordinate directly from the equation we already have.
Let us show that the metric $g$ in coordinates $(U,V)$ is diagonal, i.e., that the scalar product $g(\frac{\partial}{\partial U},\frac{\partial}{\partial V})=0$ and $g(\frac{\partial}{\partial U},\frac{\partial}{\partial U})=g(\frac{\partial}{\partial V},\frac{\partial}{\partial V})$. For this, let us note the expressions for the vector fields in coordinates $(s,t)$.

\begin{equation*}
   \frac{\partial }{\partial U}
   =s_U \frac{\partial}{\partial s} + t_U \frac{\partial}{\partial t}
   ,\qquad
   \frac{\partial }{\partial V}
   =s_V \frac{\partial}{\partial s}+ t_V \frac{\partial}{\partial t}.
\end{equation*}
Now, the scalar product $g(\frac{\partial}{\partial U},\frac{\partial}{\partial V})$ equals to
\bb
\sigma_{ss} s_U s_V+\sigma_{st}(s_V t_U+s_U t_V)+ \sigma_{tt} t_U t_V.
\label{scalar_offdiag}
\ee
Then, using~\eqref{pre-beltrami} we have
\begin{align}
\pi s_V=Y_U=\sigma_{st}s_U+\sigma_{tt}t_U,
\pi t_V=-X_U=-\sigma_{ss}s_U-\sigma_{st}t_U.
\label{eq:deriv_pi}
\end{align}
It helps us to express the derivatives with respect to $V$ through derivatives
with respect to $U$ in~\eqref{scalar_offdiag}, let us write each term
of~\eqref{scalar_offdiag} and mark by the same color the terms which cancel each
other in the sum~\eqref{scalar_offdiag}.
\begin{align*}
  \pi\sigma_{st}t_U s_V &=\mathcolor{violet}{\sigma_{st}^2s_Ut_U}+\mathcolor{blue}{\sigma_{st}\sigma_{tt}t_U^2},\\
  \pi\sigma_{st}s_U t_V &=-\mathcolor{orange}{\sigma_{ss}\sigma{st}s_U^2}-\mathcolor{violet}{\sigma_{st}^2s_U t_U}\\
  \pi\sigma_{ss}s_U s_V &=\mathcolor{orange}{\sigma_{ss}\sigma_{st}s_U^2}+\mathcolor{olive}{\sigma_{ss}\sigma_{tt}t_U s_U}\\
  \pi\sigma_{tt}t_U t_V &=-\mathcolor{olive}{\sigma_{ss}\sigma_{tt}t_Us_U}-\mathcolor{blue}{\sigma_{st}\sigma_{tt}t_U^2}
\end{align*}

As for the diagonal elements of the metric, we have

\bb
g(\frac{\partial}{\partial U},\frac{\partial}{\partial U})=\sigma_{ss}s_U^2+2s_Ut_U\sigma_{st}+\sigma_{tt}t_U^2,
\label{scalar_U}
\ee

\bb
g(\frac{\partial}{\partial V},\frac{\partial}{\partial V})=\sigma_{ss}s_V^2+2s_Vt_V\sigma_{st}+\sigma_{tt}t_V^2.
\ee

Let us repeat the same steps and express derivatives with respect to $V$, (note that we have an extra multiple $\pi$ in \eqref{eq:deriv_pi}, which appear each time we transform derivative with respect to $V$ into derivative with respect to $U$)

\begin{multline}
\frac{1}{\pi^2}(\sigma_{ss} (\sigma_{st} \mathcolor{olive}{s_U}+\sigma_{tt}\mathcolor{violet}{t_U})^2\\
+2\sigma_{st}(\sigma_{st} \mathcolor{olive}{s_U}+\sigma_{tt}\mathcolor{violet}{t_U})(-\sigma_{ss}s_U-\sigma_{st} t_U)\\
+\sigma_{tt}(\sigma_{ss}\mathcolor{olive}{s_U}+\sigma_{st}\mathcolor{violet}{t_U})^2).
\label{scalar_V}
\end{multline}
Then, comparison of coefficients in front of $s_U^2$ shows that in \eqref{scalar_U} we have $\sigma_{ss}$, while in \eqref{scalar_V} we have
\bb
\frac{1}{\pi^2}\sigma_{ss}(\sigma_{st}^2-2\sigma_{st}^2+\sigma_{tt}\sigma_{ss})=
\sigma_{ss}\frac{\sigma_{ss}\sigma_{tt}-\sigma_{st}^2}{\pi^2}=\sigma_{ss}
\ee
(terms which contribute to the coefficient are in the olive color).
A similar computation done for coefficients in front of $t_U^2$ shows that (the terms are in the red color), on one hand, we get $\sigma_{tt}$, while on the other hand
\bb
\frac{1}{\pi^2}\sigma_{tt}(\sigma_{ss}\sigma_{tt}-\sigma_{st}^2)=\sigma_{tt}.
\ee
Finally, the same computation shows agreement for the coefficient in front of $s_Ut_U$,
$2\sigma_{s,t}$ on one hand, and $\frac{1}{\pi^2}(2\sigma_{s,t}(\sigma_{s,s}\sigma_{t,t}-\sigma_{s,t}^2-\sigma_{t,t}\sigma_{s,s}+\sigma_{t,t}\sigma_{s,s}))=2\sigma_{s,t}$ on the other.

\subsection{An intermediate parameterization}
  \label{sec:param_u}
  It may happen that there are several points in the liquid region with the same
  slope $(s,t)\in\mathring{\mathcal{N}}\setminus\mathscr{G}$. In that case it is
  not possible to parameterize the liquid region by the slope, and thus by the
  coordinate $z$.
  However, the degree $d$ of $\nabla \mathfrak{h}\colon
  \mathcal{L}\rightarrow\mathring{\mathcal{N}}\setminus\mathscr{G}$ is constant
  except many at a finite number of points.
  We can therefore introduce a ramified covering $\Sigma$ of degree $d$ over the
  spectral curve $\{P(z,w)=0\}$ such that the variables $z$ and $w$ are now meromorphic
  functions of the variable $u\in\Sigma$. Let $\Sigma^+$ the preimage by $z$ of
  the upper-half plane. This gives a diffeomorphism from $\Sigma^+$ to the
  liquid region.

  The functions $s$ and $t$, considered before as functions of $z$, can be now
  seen also as functions of $u\in\Sigma^+$.

\subsection{The intercept function}
Let us look at the graph of the minimizer $\mathfrak{h}^{\star}$. More precisely, take
the tangent plane $\mathcal{P}_{x_0,y_0}$ to it at a given point $(x_0,y_0)$. It
has a slope given by $\nabla \mathfrak{h}^{\star}(x_0,y_0)=(s(x_0,y_0),t(x_0,y_0))$. It
also intersects the vertical axis at the
ordinate
$(\mathfrak{h}^{\star}-(sx+ty))$. The
function $c:=\mathfrak{h}^{\star}-sx-ty$ is called \textit{the intercept}. With
its help, we can parametrize the tangent planes to the plot of $\mathfrak{h}^{\star}$ as follows,

\begin{equation}
    \mathcal{P}_{x_0,y_0}=\{(x,y,z)\in\RR^3|\{s(x_0,y_0) x +t(x_0,y_0) y +c(x_0,y_0)=z\}.
    \label{tangent_planes}
\end{equation}

The important properties of $c$ for us are first, it is constant on every frozen
region since on each frozen region $\mathfrak{h}^{\star}$ is linear with a fixed
slope, and so is its tangent plane. Thus, after subtracting the linear part of
$\mathfrak{h}^{\star}$, the difference takes a constant value on each frozen region. Further, the minimizer $\mathfrak{h}^{\star}$ can be reconstructed from $s$, $t$ and $c$ as $\mathfrak{h}^{\star}=sx+ty+c$.
After it, the limit shape is obtained as the envelope of planes given
by~\eqref{tangent_planes}. The main property of these
planes is harmonicity of functions $s$, $t$ and $c$ as functions of the conformal
coordinate, which is discussed in the next section.

\subsection{\texorpdfstring{Harmonicity of $s,t$ and $c$}{Harmonicity of s, t and c}}
The property of functions $s,t$ and $c$ as functions of conformal coordinate
$u$ is that they are harmonic by Theorem~3.1 from~\cite{KP:plane}.


Now, let us re-formulate \hyperref[prop2.1]{Proposition~\ref{prop2.1}} in real terms and prove harmonicity of $s,t$ and $c$ in conformal coordinates.
\begin{proposition}
    The functions $s$, $t$ and $c$ are harmonic in the variable $u\in\Sigma^+$.
    \label{harmonicity}
\end{proposition}

\begin{proof}
Functions $s$ and $t$ are harmonic in the variable $u$ since they are equal to
the argument (imaginary part of the logarithm) of $z$ and $w$, which are
holomorphic functions of $u$ in the interior of $\Sigma^+$.


Next, consider 
$c=(\mathfrak{h}-(sx+ty))$,
and let us apply to it
$\partial_u$
and
$\partial_{\bar{u}}$
remembering that 
\begin{equation*}
  \nabla \mathfrak{h}=(\mathfrak{h}_x,\mathfrak{h}_y)=(s,t),
\end{equation*}
\bb
c_u=(\mathfrak{h}-(sx+ty))_u=
\mathfrak{h}_x x_u+
\mathfrak{h}_y y_u
-(s_u x+ t_u y+ s x_u + t y_u)=-s_u x - t_u y.
\ee

Then, applying
$\partial_{\bar{u}}$
we obtain by harmonicity of $s,t$

\bb
%
\partial_{\bar{u}}(s_u x + t_u y)= s_u x_{\bar{u}}+t_u y_{\bar{u}}.
\ee

Moreover, taking the derivative of $u=u(x,y)$ with respect to $\bar{u}$
(which gives 0), we get by the chain rule:
$0=u_x x_{\bar{u}} + u_y y_{\bar{u}}$. Since $u$ is locally a holomorphic
function of $z$, it is also an intrinsic coordinate. We can therefore
apply~Proposition~\ref{intrinsic}, and get
\begin{equation*}
  -\frac{y_{\bar u}}{x_{\bar u}} = \frac{u_x}{u_y} = \frac{s_u}{t_u},
\end{equation*}
from which it follows that
\begin{equation*}
  \Delta_u c = 4 c_{u,\bar{u}} = 0.
\end{equation*}
\end{proof}

Therefore, we can hope to reconstruct all three functions as long as we are able
to perform explicit harmonic extensions of piecewise constant boundary
conditions. As a corollary, we can obtain the limit shape $\mathfrak{h}^{\star}$
as envelope of harmonically moving planes by Theorem \eqref{limit_shape_envelope}, see also \cite[Theorem~3.2]{KP:plane},

\subsection{Tangent plane equation}

Introduce the main equation of the method, Equation~(20) from~\cite{KP:plane},
which will help us to compute the arctic curve.

\begin{proposition}
  Inside the liquid region, parameterized by $u\in\Sigma^+$, the following
  equation holds:
    \bb
    s_u x+t_u y +c_u=0
    \ee
    \label{tangent_eq}
\end{proposition}

What is important is that this complex equation is equivalent to two real
equations for the real and imaginary parts, which gives for every $u\in\Sigma^+$
a linear system for
$(x(u),y(u))$, once we know $s$, $t$, and $c$ as functions of $u$.

\begin{proof}
    Let us apply the Wirtinger derivative $\partial_u $ to the minimizer $\mathfrak{h}$, remembering that $\nabla \mathfrak{h}=(s,t)$, we have the following equalities
\bb
\mathfrak{h}_u=sx_u+ty_u=(sx+ty)_u-(s_u x+t_u y)
\ee

\bb
s_u x+ t_u y +(\mathfrak{h}-(sx+ty))_u=0,
\ee
which is the same as 
\bb
s_u x +t_u y +c_u=0.\qedhere
\ee
\end{proof}

In practice, we deal with harmonic extensions made of linear combinations of $\arg (u)$, the important identities are first $\arg (u)=\Im \log (u)$, which we use to take the Wirtinger derivative as
\bb
\partial_u \arg (u)=\frac{1}{2i} \frac{1}{u}.
\ee

As a corollary of \eqref{harmonicity} and \eqref{tangent_eq}, we obtain the plot of the minimizer $\mathfrak{h}^\star$ over $\Sigma^+$ can be recovered as an envelope of harmonically moving planes with the slope given by $(s(u),t(u))$, $u\in\Sigma^+$.

\begin{theorem}[Theorem 3.2 \cite{KP:plane}]
The graph of the minimizer $\mathfrak{h}^{\star}$ over $u\in\Sigma^+$ is the envelope of harmonically moving planes $\mathcal{P}_u$ that satisfy two conditions:
\begin{itemize}
    \item $\mathcal{P}_u:=\{(x(u),y(u),z(u))\in\RR^3|x s(u)+ y t(u)+c(u)=z(u)\}$
    \item $s_u x +t_u y +c_u=0$
\end{itemize}
    \label{limit_shape_envelope}
    Here, $x$ and $y$ are functions of $u$.
\end{theorem}

\section{Applications for the random domino tilings}
\label{sect:3}
In the rest of the chapter, we are going to explain the method applied to a simply-connected domain, the Aztec diamond, and to our main example, the Aztec diamond with a hole.
We also need to use assumptions for the number of frozen regions of the particular domains, and values of $s,t$ and $c$ there. This data is inferred from discrete boundary conditions and computer simulations.

\begin{figure}[h!]
    \centering
    \includegraphics[width=0.3\linewidth]{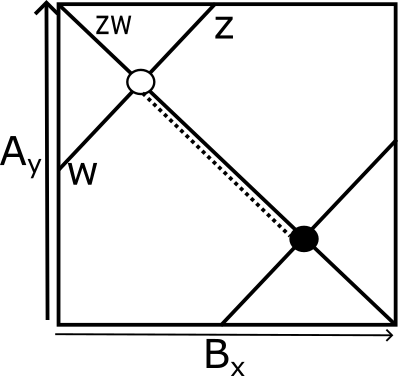}
    \caption{Weights of the fundamental domain of the square grid with $A_y$ and $B_x$ cycles. The reference dimer cover is in dashed color.}
    \label{fig:fundamental_weights}
\end{figure}

For the square grid, we can take the fundamental domain represented
on Figure~\ref{fig:fundamental_weights}, the corresponding Kasteleyn operator is given by
\bb
\mathcal{K}(z,w)=1+z+w-zw.
\ee

\subsection{A case with a simply-connected liquid region: the uniform Aztec
diamond}

Recall that the Aztec diamond of order $N$ is the union of unit squares $S(m,n)$ of the square lattice whose centers $(m,n)$ satisfy
\bb
|m-\frac{1}{2}|+|n-\frac{1}{2}|\leq N.
\ee
It is also convenient to introduce chess-board coloring on the square grid, this
way we obtain $4$ types of dominoes, see Figure~\ref{fig:aztec_color}.
\begin{figure}[t!]
    \centering
    \begin{subfigure}{0.2\linewidth}
    \centering
    \hfill
    \includegraphics[width=\linewidth]{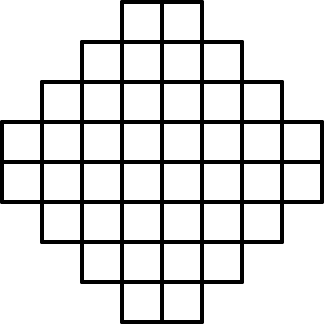}
    \end{subfigure}
    \hfill
    \begin{subfigure}{0.35\linewidth}
        \includegraphics[width=\linewidth]{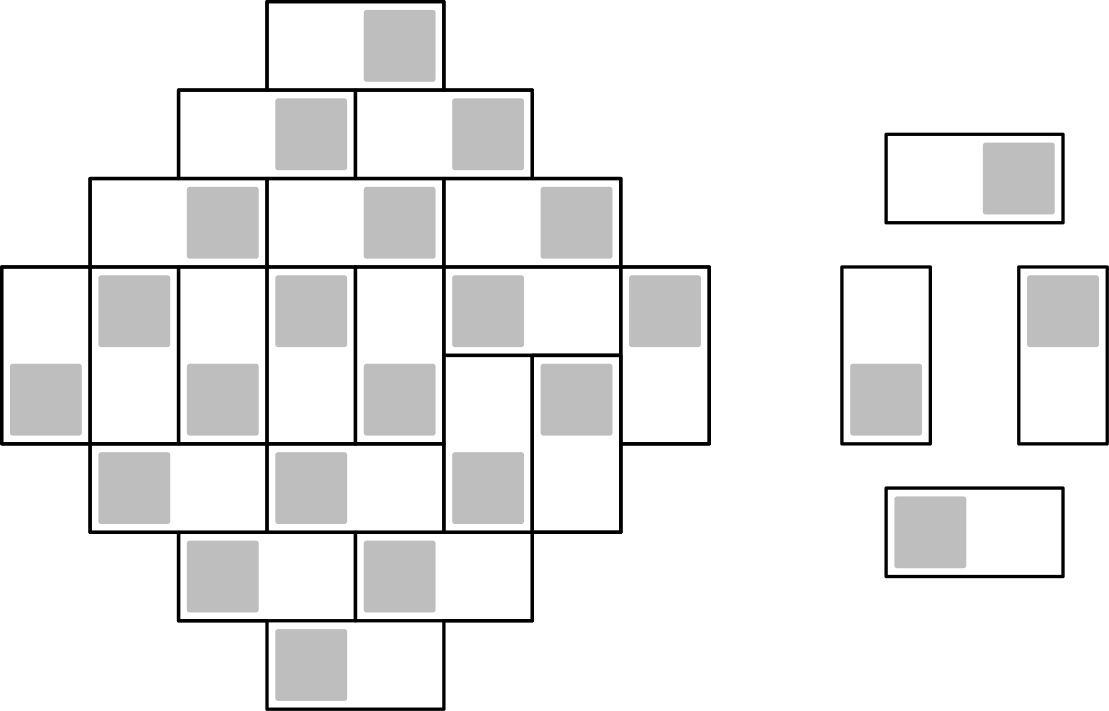}
    \end{subfigure}
    \hfill\mbox{}
    \caption{Example of the Aztec diamond of order $4$ on the left, and a domino
    tiling of it on the right, with four types of dominoes according to the chess-board coloring.}
\end{figure}

\begin{figure}[h!]
    \centering
    \hfill
    \begin{subfigure}{0.4\linewidth}
    \centering
\includegraphics[width=\linewidth]{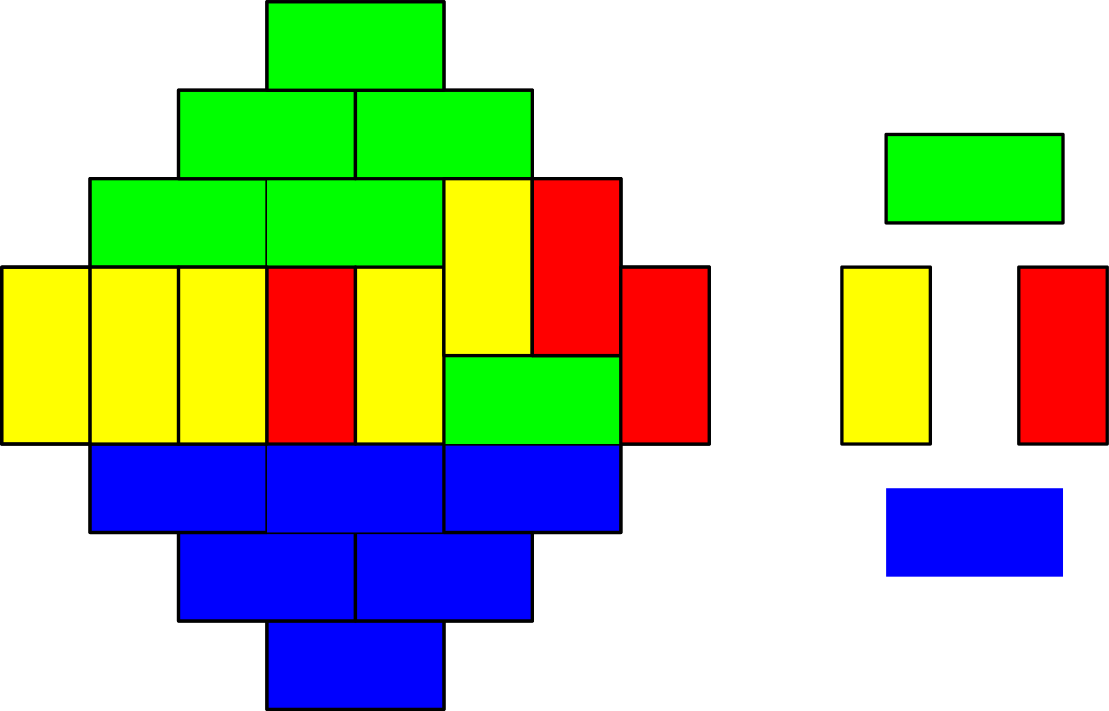}
    \end{subfigure}
    \hfill
    \begin{subfigure}{0.35\linewidth}
        \includegraphics[width=\linewidth]{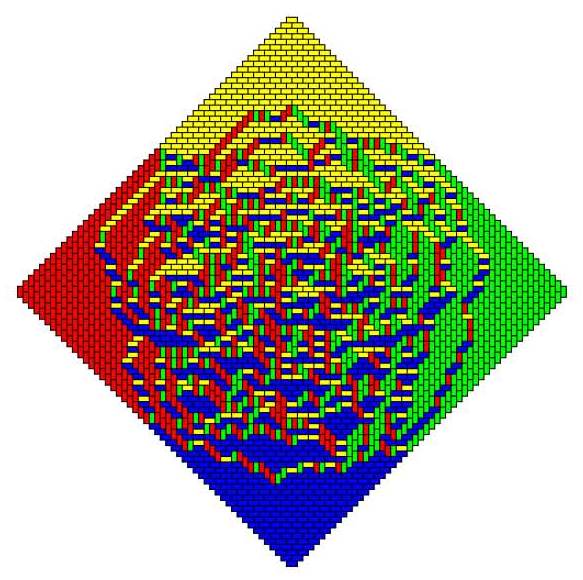}
    \end{subfigure}
    \hfill\mbox{}
    \caption{Example of the same domino tiling in color notation, and domino tiling of Aztec diamond of order $50$ by J.Propp.}
    \label{fig:aztec_color}
\end{figure}

We consider in this section the case of the uniform measure on tilings of
the Aztec diamond of size $N$.
As $N$ goes to infinity, a random height function converges in probability to a
  deterministic function which is linear in the four regions in the limiting
  renormalized square deprived of the inscribed disc (the \emph{frozen
  regions}) and is smooth inside the inscribed disc (the \emph{liquid} region). This
  statement is due to Jockush, Propp and Shore~\cite{JPS}.
  The goal of this section is to recover this result using the tangent plane
  method, described above, following loosely~\cite[Section 6.1]{KP:plane},[Section 2.3]\cite{KP:new}.

\subsubsection{Applying the tangent plane method}

In the case of the uniform Aztec diamond, it is convenient to use a coordinate
system $(x,y)$ in the renormalized domain wich is rotated by 45 degrees with respect to
the horizontal/vertical discrete coordinate axes, so that the renormalized
domain becomes in the limit the unit square $[0,1]\times[0,1]$. See
Figure~\ref{fig:tangent_planes_boundary_aztec} to see the directions of the two
axes for $x$ and $y$.

 In the liquid region $\mathcal{L}$,
there is a unique point where the limiting height function has a given slope in
the interior of the Newton polygon (the fact that each of the four frozen phases
is seen only once on the boundary of $\mathcal{L}$ is a sign that the degree $d$
from $u$ to $z$ or $w$ is 1). This is thus a case when we
can take $u=z$: the liquid region can be paramaterized by
$z\in\mathbb{H}$. We now determine the values for $s$, $t$ and $c$ on
the boundary of the liquid region (which corresponds to
$z\in\mathbb{R}\cup\{\infty\}$, bounded by the four distinct frozen regions.

Each interval $(-\infty,-1)$, $(-1,0)$, $(0,1)$, $(1,\infty)$ corresponds to an
arc of the arctic curve touching a given frozen phase, with slope $(s,t)$ given
respectively (with the convention of the slope given by
Equation~\eqref{pre-beltrami0}) by $(1,0)$, $(0,0)$, $(0,1)$ and $(1,1)$
respectively.

We form a table consisting of four columns for each frozen region, and three
rows for each function $s$, $t$ and $c$.
\begin{center}
\begin{tabular}{c | c | c | c | c}
&  (1) & (2) & (3) & (4) \\
&$z<-1$&$-1<z<0$&$0<z<1$&$1<z$\\
&{\scriptsize $1>w>0$}& {\scriptsize $0>w> -1$}
& {\scriptsize $-1>w$} & {\scriptsize $w>1$}\\\hline
$s$ & $1$ & $0$ & $0$ & $1$\\
$t$ & $0$ & $0$ & $1$ & $1$\\
$c$ & $0$ & $1$& $0$ & $0$
\end{tabular}\end{center}

The value of $c$ is determined up to some additive constant, fixed here to be
zero. Its value in each region is fixed by the condition that the linear parts
of the height functions with the given slopes in the table should match at the
``turning points'', where two frozen regions touch on the boundary of the Aztec
diamond. See Figure~\ref{fig:height_boundary_aztec} for a picture of those
linear pieces.

\begin{figure}[htpb]
  \centering
  \includegraphics[width=6cm]{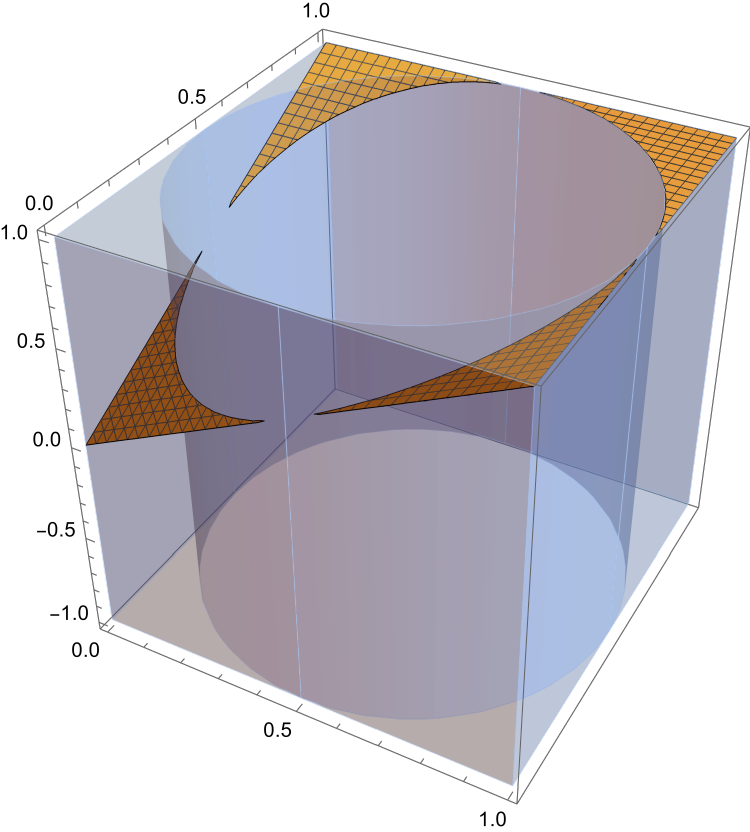}
  \caption{Height above the frozen regions of the Aztec diamond, which is a
  piecewise linear function. The values of $c$ on each arc of the arctic curve
  reflect the continuity of the height along the boundary of the liquid region.}
  \label{fig:height_boundary_aztec}
\end{figure}

We represent visually on Figure~\ref{fig:tangent_planes_boundary_aztec} the
information of the table, and four specific values of $z$ at the transition
between two frozen phases near the arctic curve.

\begin{figure}[htpb]
  \centering
  \def\svgwidth{10cm}
  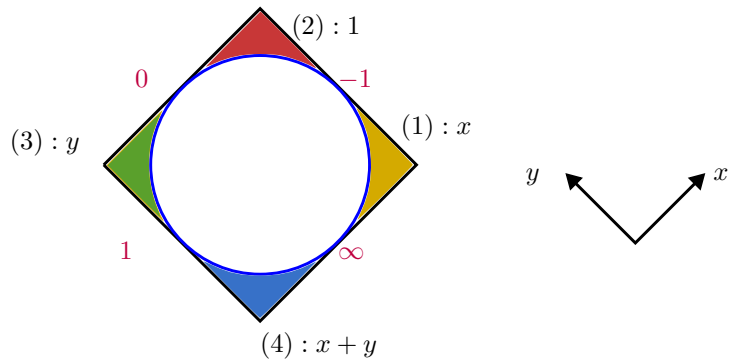
  \caption{The four frozen phases of the Aztec diamond, with colors
    corresponding to branches of the boundary of the amoeba from
    Figure~\ref{fig:amoeba}. For each region, we indicate in black the equation
    for the ordinate as an affine function of $x$ and $y$. For each ``turning
    point'', we indicate in purple the corresponding value of $z$ in the
    parametrization.
  }
  \label{fig:tangent_planes_boundary_aztec}
\end{figure}



\subsubsection{Parametrization of the limit shape of the Aztec diamond}
\label{sect_aztec}

We know that $s$, $t$ and $c$ are harmonic functions of the variable $z$. We
construct explicit harmonic extensions of the boundary values from the table.
For $s$ and $t$, the answer is given directly by Equation~\eqref{pre-beltrami0},
which can be rewritten as
\begin{equation*}
  s(z) = \frac{1}{\pi}\arg (-w(z)) = \frac{1}{\pi}\arg \frac{z+1}{z-1},
  \qquad
  t(z) = -\frac{1}{\pi}\arg (-z) = 1-\frac{1}{\pi}\arg(z).
\end{equation*}
For $c$, we proceed with the same idea, using building blocks of the form
\begin{equation*}
f(z)=\frac{1}{\pi} \arg \frac{z-b}{z-a}
\end{equation*}
which is the harmonic extension of the indicator function of interval $[a,b]$.
The corresponding harmonic extensions is:
\begin{equation*}
  c(z)= \frac{1}{\pi}(-\pi+\arg(z-1)).
\end{equation*}




Another choice of conformal coordinate instead of $z$ which will be useful in
connection with the next example is choosing $\zeta=\frac{2}{\pi}\arctan z$,
which maps the complex plane to an infinite cylinder $\mathbb{C}/2\mathbb{Z}$.
The upper-half plane for $z$ parametrizing the liquid region (respectively the
real axis together with the point at infinity parametrizing the arctic curve)
is mapped to the upper half of the cylinder (respectively to
$\mathbb{R}/2\mathbb{Z}$).



The table of boundary conditions of $s,t$ and $c$ for $\mathcal{AD}$ for the
variable $\zeta$ is the following (where the intervals for $\zeta$ represent a
cyclic order, as on the cylinder, $-1=1$):
\begin{center}
\begin{tabular}{c | c | c | c | c}
&  (1) & (2) & (3) & (4) \\
&$-1<\zeta< -\frac{1}{2}$ & $-\frac{1}{2} < \zeta <0$ & $0< \zeta<\frac{1}{2}$ &
$ \frac{1}{2}<\zeta<1$\\\hline
$s$ & $1$ & $0$ & $0$ & $1$\\
$t$ & $0$ & $0$ & $1$ & $1$\\
$c$ & $0$ & $1$& $0$ & $0$
\label{table_aztec_trig}
\end{tabular}\end{center}

The harmonic extensions in this parametrization are given by the pullback under the map
$z\mapsto \frac{2}{\pi}\arctan(z)$,
that is,
by a composition with the map
$\zeta\mapsto \tan (\frac{\pi\zeta}{2})$.

\begin{align*}
  s(\zeta) & = \frac{1}{\pi}\arg
  \frac{1-\tan(\frac{\pi\zeta}{2})}{1+\tan(\frac{\pi\zeta}{2})}=\frac{1}{\pi}\arg\tan(\frac{\pi}{2}(\frac{1}{2}-\zeta)), \\
  t(\zeta) & = 1-\frac{1}{\pi}\arg(\tan(\frac{\pi\zeta}{2})), \\
  c(\zeta) & = \frac{1}{\pi}\arg\frac{\tan(\frac{\pi}{2}\zeta)}{\tan(\frac{\pi}{2}\zeta)+1}.
\end{align*}




Now, we need to plug them into linear system for $(x,y)$ given by \hyperref[tangent_eq]{Equation (\ref{tangent_eq})}.
For it, we need derivatives of $s,t$ and $c$. Recall that we differentiate the
$\arg z$ by applying the Wirtinger derivative using the identity $\arg (z)=\Im
\log z$, which gives a factor $\frac{i}{2}$ to each $\arg $, and therefore
factorizes. Thus, we present derivatives without this common factor.
In the coordinate $z$ after a common multiplication by $2\pi i$, we have

\begin{align*}
  2 i\pi s_z & = \frac{1}{z-1}-\frac{1}{z+1}=\frac{2}{z^2-1},\\
2 i \pi t_z &  = \frac{1}{z},\\
2 i\pi c_z &   = \frac{1}{z}-\frac{1}{z+1}=\frac{1}{z(z+1)}.
\end{align*}







With the help of expressions, we can build parametric plots using linear system $s_z x+ t_z y +c_z=0$
for finding $(x(z),y(z))$,
which continuously depends on $z$, and then by inversion, compute for every
$(x,y)$ in the liquid region the value of $s(x,y)$, $t(x,y)$, $c(x,y)$, which
would give the equation of the tangent plane to the limit shape above the point
$(x,y)$. This would finally allow us to reconstruct the limiting height function
$\mathfrak{h}$ as the enveloppe of this family of tangent planes.

But before that, there is something simpler we can do: we can look at the image
of $\partial\mathbb{H}$ by the map $z\mapsto (x(z), y(z))$ which corresponds
exactly to the arctic curve, that is the boundary of the liquid region.


Furthermore, from a computational point
of view, it may be hard to give an analytical expression for the solution
$(x(z),y(z))$ (except for the uniform Aztec diamond in coordinate $z$, as we
will see). Rather, we can instead for each $z$
compute the linear system and its solution, $(x(z),y(z))$ and plot it by varying
$z$. 

In fact, the functions defining the coefficients may not well-defined precisely on the
boundary, thus if we try to do it numerically, we add a small positive imaginary
part to $z$. This
reflects the fact that
gradient $\nabla \mathfrak{h}$ is defined only in the interior of the liquid
region, and not on the arctic curve. 
\begin{equation}
    \Re(s_z) x +\Re (t_z) y +\Re(c_z)=0, 
    \label{eq:criteq_real}
\end{equation}
\begin{equation}
    \Im(s_z) x +\Im(t_z) y +\Im(c_z)=0.  
    \label{eq:criteq_imag}
\end{equation}
This is the approach we will use for the next examples.


But it turns out, as it is often the case for the Aztec diamond, that
computations are quite easy, and we can find $x(z), y(z)$ for every
$z\in\mathbb{H}$ by solving the simple linear system,
\begin{equation}
  x(z)= \frac{|z-1|^2}{2(|z|^2+1)},\qquad
  y(z) = \frac{1}{1+|z|^2},
  \label{eq:param_liq_aztec_unif}
\end{equation}
or even invert it to
find $z(x,y)$. Indeed, the two real equations~\eqref{eq:criteq_real}
and~\eqref{eq:criteq_imag} are equivalent to the degree-2 complex equation in
$z$:
\begin{equation*}
  2 z x +(z^2-1)y-z+1=0,
\end{equation*}
for which we search the solution with positive (or rather non-negative) root,
which is given by:
\begin{equation*}
  z(x,y) = \frac{1-2x+i\sqrt{1-(2x-1)^2-(2y-1)^2}}{2y},
\end{equation*}



At this point, we already recognize the arctic circle of radius
$\frac{1}{2}$ and center $(\frac{1}{2},\frac{1}{2})$ as the limit of definition
of the square root defining $z$: this corresponds to the boundary of the liquid
region. As mentionned before, another way to recover a parametrization of the
arctic curve would be to take the solutions from
Equation~\eqref{eq:param_liq_aztec_unif} and plot it for
$z\in\mathbb{R}\cup\{\infty\}$. See
Figure~\ref{fig:phase_z_aztec_unif} for a plot of the argument of $z$ across the
liquid region, and for a parametric plot of the arctic curve.

\begin{figure}
  \hfill
  \includegraphics[width=5cm]{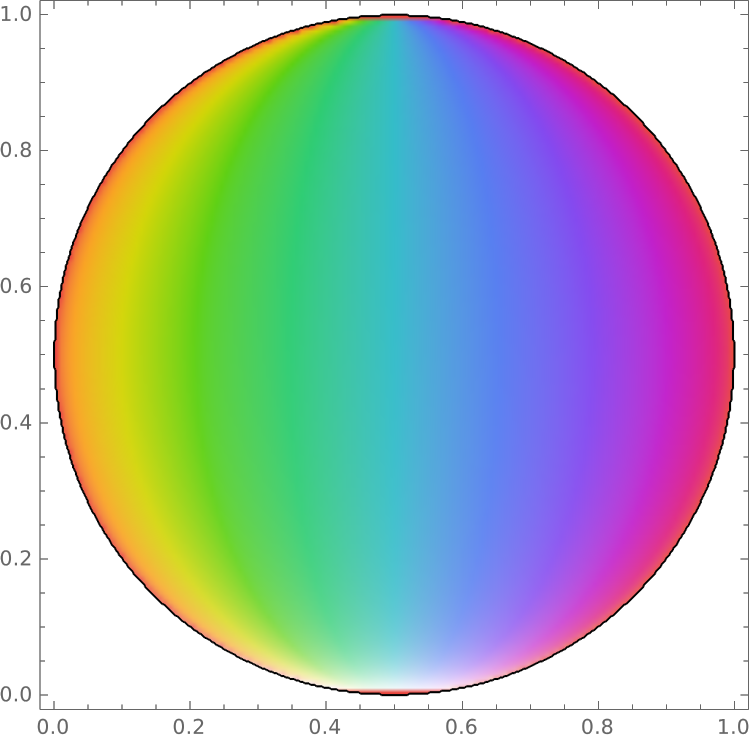}\hfill
  \includegraphics[width=5cm]{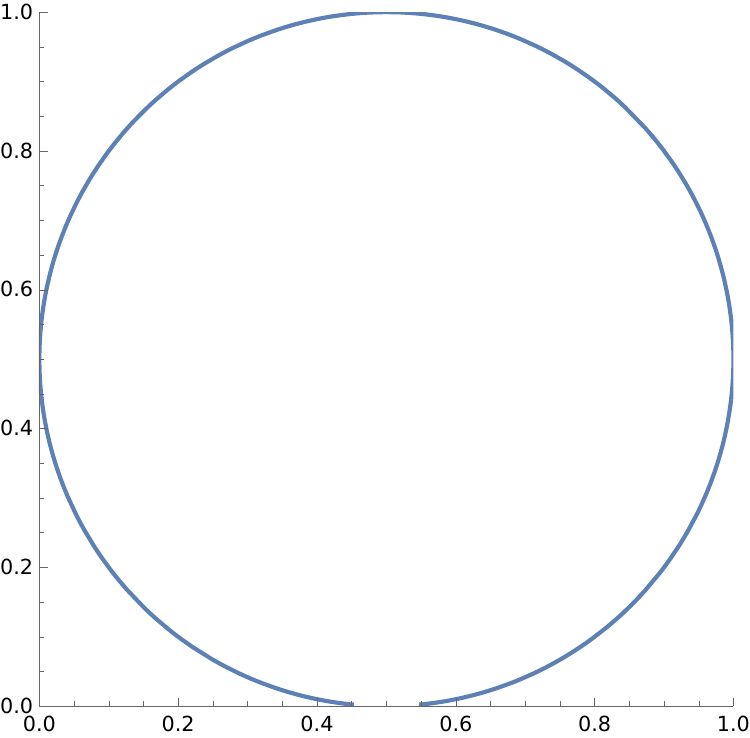}
  \hfill\mbox{}
  \caption{Left: the argument of $z(x,y)$ for $(x,y)$ in the liquid region. Cyan
    means that $z$ is pure imaginary, purple and pink (close to the right
    boundary) correspond to $z$ near $\mathbb{R}^-$, whereas yellow and orange (on
    the left) means that $z$ is close to $\mathbb{R}^+$.
  Right: a large piece of the arctic curve obtained from the parametrization by
Equations~\ref{eq:param_liq_aztec_unif}, for $z$ real between $-20$, $20$. One
can see that the bottom of the circle is cropped, meaning that this part
corresponds to $z$ in a neigbourhood of $\infty$.
}\label{fig:phase_z_aztec_unif}
\end{figure}

\begin{remark}
  Instead of taking the smallest fundamental domain with one white and one black
  vertex, and considering the spectral curve for it, one could have taken a fundamental domain made
  of two white and two black vertices forming a $2\times 2$-square.
  See Figure~\ref{fig:fund_domain_z2_2x2}. 
  The coordinate axes
  $(x,y)$ are now aligned with the horizontal and vertical axes of the
  square lattice, so that 
  the renormalized domain is a rotated square $|x|+|y|\leq 1$.
  If we label with 1 the vertices on the top row, and with 2 those of the second
  row of the fundamental domain, the modified Kasteley matrix $\mathcal{K}(z,w)$
  is given by
  \begin{equation}
    \mathcal{K}(z,w)=
    \begin{pmatrix}
      1-w & z-1 \\
      \frac{1}{z}-1 & \frac{1}{w}-1
    \end{pmatrix}
    \label{eq:kast_2x2_unif}
  \end{equation}
  and the characteristic polynomial can be written as
  $-4+z+w+\frac{1}{z}+\frac{1}{w}$. See Figure~\ref{fig:fund_domain_z2_2x2},
  right for the Newton polygon of this characteristic polynomial, and
  Figure~\ref{fig:amoeba_z2_2x2} for its amoeba, with the indication of the four
  frozen phases.
  A similar table for boundary values of
  $s$, $t$ and $c$.
\begin{center}
\begin{tabular}{c | c | c | c | c}
&  (1) & (2) & (3) & (4) \\\hline
$s$ & $1$ & $0$ & $-1$ & $0$\\
$t$ & $0$ & $-1$ & $0$ & $1$\\
$c$ & $-\frac{1}{2}$ & $\frac{1}{2}$& $-\frac{1}{2}$ & $\frac{1}{2}$
\end{tabular}\end{center}
Note that in that case, $z$ itself does not parametrize the liquid region: for a
given value of $z$, there are two values of $w$ such that $(z,w)$ is on the
spectral curve. However, one can parametrize both $z$ and and $w$ using rational
fractions of a variable $u$ living on the Riemann sphere, which could be chosen
as the conformal coordinate.
\end{remark}

\begin{figure}
  \centering
  \hfill{}
  \def\svgwidth{4cm}
  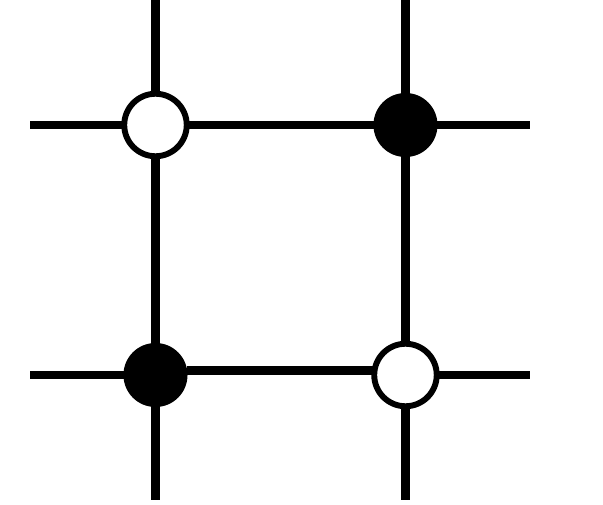
  \hfill
  \def\svgwidth{3cm}
\begingroup%
  \makeatletter%
  \providecommand\color[2][]{%
    \errmessage{(Inkscape) Color is used for the text in Inkscape, but the package 'color.sty' is not loaded}%
    \renewcommand\color[2][]{}%
  }%
  \providecommand\transparent[1]{%
    \errmessage{(Inkscape) Transparency is used (non-zero) for the text in Inkscape, but the package 'transparent.sty' is not loaded}%
    \renewcommand\transparent[1]{}%
  }%
  \providecommand\rotatebox[2]{#2}%
  \newcommand*\fsize{\dimexpr\f@size pt\relax}%
  \newcommand*\lineheight[1]{\fontsize{\fsize}{#1\fsize}\selectfont}%
  \ifx\svgwidth\undefined%
    \setlength{\unitlength}{192.1259565bp}%
    \ifx\svgscale\undefined%
      \relax%
    \else%
      \setlength{\unitlength}{\unitlength * \real{\svgscale}}%
    \fi%
  \else%
    \setlength{\unitlength}{\svgwidth}%
  \fi%
  \global\let\svgwidth\undefined%
  \global\let\svgscale\undefined%
  \makeatother%
  \begin{picture}(1,0.93950849)%
    \lineheight{1}%
    \setlength\tabcolsep{0pt}%
    \put(0,0){\includegraphics[width=\unitlength,page=1]{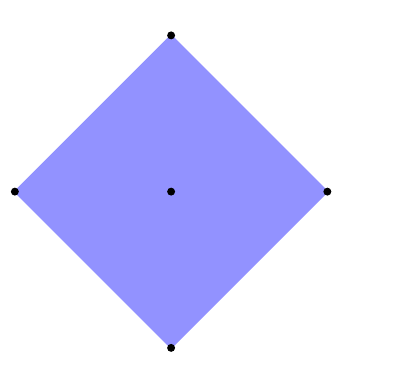}}%
    \put(0.40796674,0.89015029){\color[rgb]{0,0,0}\makebox(0,0)[lt]{\lineheight{1.25}\smash{\begin{tabular}[t]{l}$(4)$\end{tabular}}}}%
    \put(0.40796674,0.01182036){\color[rgb]{0,0,0}\makebox(0,0)[lt]{\lineheight{1.25}\smash{\begin{tabular}[t]{l}$(2)$\end{tabular}}}}%
    \put(0.85689091,0.46074453){\color[rgb]{0,0,0}\makebox(0,0)[lt]{\lineheight{1.25}\smash{\begin{tabular}[t]{l}$(1)$\end{tabular}}}}%
    \put(-0.00192061,0.48026305){\color[rgb]{0,0,0}\makebox(0,0)[lt]{\lineheight{1.25}\smash{\begin{tabular}[t]{l}$(3)$\end{tabular}}}}%
  \end{picture}%
\endgroup%

  \hfill\mbox{}
  \caption{Left: a $2\times 2$-fundamental domain for the square lattice. Kasteleyn
    minus signs are indicated on the edges, as well as the extra factors
    $z^{\pm 1}$ and $w^{\pm 1}$ for edges crossing the fundamental domain.
    Right: the Newton polygon of the corresponding characteristic polynomial.
  }
  \label{fig:fund_domain_z2_2x2}
\end{figure}

\begin{figure}
  \centering
  \def\svgwidth{6cm}
  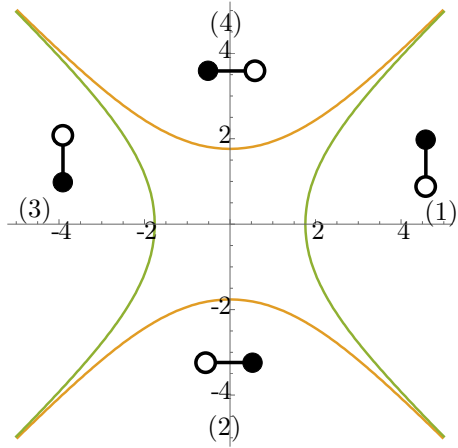
  \caption{Amoeba of the characteristic polynomial
  $-4+z+w+\frac{1}{z}+\frac{1}{w}$. The connected components of the complement
are labeled with numbers corresponding to frozen phases and vertices of the
Newton polygon, together with the corresponding type of edges.}
  \label{fig:amoeba_z2_2x2}
\end{figure}


\subsection{A case with a multiply-connected liquid region, Aztec diamond
with two-periodic weights}

\label{doubly_aztec}
We follow loosely the discussion here~\cite[Section 2.3]{KP:new}.

Before analysis of a multiply-connected region, let us discuss our example with
a multiply-connected liquid region, where the region is the simply-connected,
but the liquid region is non simply-connected due to presence of smooth or gas
phase, so-called bubble. This can be achieved for non-uniform distribution on
the set of dimer configurations, and the most common example of such phenomenon
is the \textit{doubly-periodic Aztec diamond}\cite{Sunil_doubly,KP:new}.

\begin{figure}
\centering
\hfill
\begin{subfigure}{0.3\linewidth}
    \includegraphics[width=\linewidth]{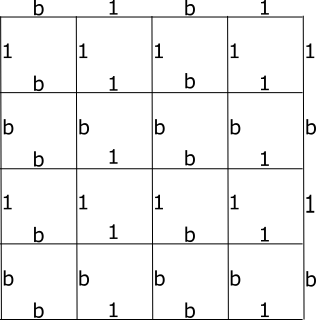}
\end{subfigure}
\hfill
\begin{subfigure}{0.3\linewidth}
    \includegraphics[width=\linewidth]{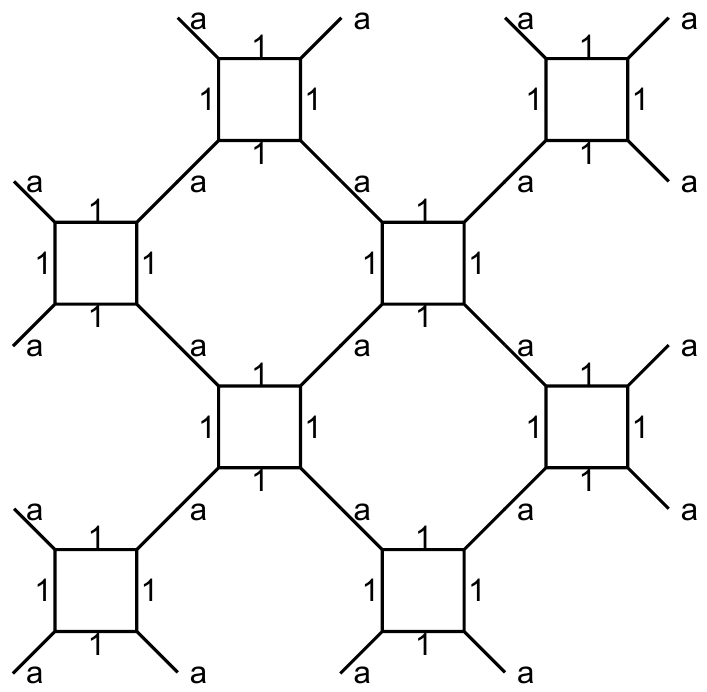}
\end{subfigure}
\hfill\mbox{}
    \caption{Edge weights of the square grid on the left, and the corresponding weights of the octagon lattice on the right, the correspondence between weights is $a^2=2b$.}
    \label{fig:octagon}
\end{figure}

This system of weights for the square lattice is related to the dimer model on
the square-octagon graph: there is a
weight-preserving correspondence between configurations on the two lattices
(obtained by performing the so-called \emph{spider move} or \emph{square move}),
which maps domino tilings of doubly-periodic square grid
to dimer configurations on the square-octagon lattice as long as the
weights $a^2=2b$ from \hyperref[fig:octagon]{Figure \ref{fig:octagon}}.

In that case, the modified Kasteleyn matrix from~\eqref{eq:kast_2x2_unif} becomes
\begin{equation}
  \label{eq:kast_2x2_2perio}
    \mathcal{K}(z,w)=
    \begin{pmatrix}
      b-w & z-b \\
      \frac{1}{z}-b & \frac{1}{w}-b
    \end{pmatrix}
\end{equation}
and the characteristic polynomial is $-2(1+b^2)+b(z+w+\frac{1}{z}+\frac{1}{w})$.
See Figure~\ref{fig:amoeba_2x2_perio} for its amoeba. The spectral curve is a
genus~1 algebraic curve $\mathcal{C}$ if $b\neq 1$, which is a Harnack curve. It can be uniformized
to a rectangular torus $\mathbb{C}/(2\mathbb{Z}+2\tau\mathbb{Z})$, with $\tau\in
i\mathbb{R}^*_+$: there is a birational map
\begin{equation*}
  \psi:\zeta\in\text{torus}\mapsto (z(\zeta), w(\zeta))\in \mathcal{C}.
\end{equation*}
See for example~\cite[Section~5]{BCdT:ellipt}
and~\cite[Proposition~9]{boutillier2024focks}.

The relation between $b$ and
$\tau$ (assuming that $b<1$, there is a symmetry $b\leftrightarrow \frac{1}{b}$)
is given by
\begin{equation*}
  b=\sqrt{k'},\qquad \text{with\ } k'=\frac{\theta_4(0|\tau)^2}{\theta_3(0|\tau)^2}.
\end{equation*}
In particular, for $\tau=1$, $k'=\frac{1}{\sqrt{2}}$, and $b=2^{-1/4}$.

\begin{figure}
  \centering
  \includegraphics[width=5cm]{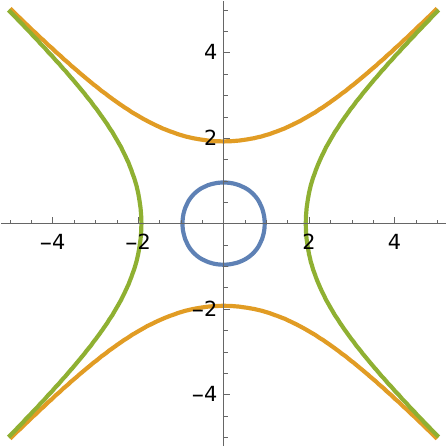}
  \caption{Amoeba of the 2-periodic weight function on the square lattice, for
    $b=0.5$. The hole in the middle corresponds to the \emph{gas} or
    \emph{smooth phase} of the corresponding dimer model, which describes the
    bubble appearing in the middle of the pictures of
    Figure~\ref{fig:simul_doubly}.
  }
  \label{fig:amoeba_2x2_perio}
\end{figure}

\begin{figure}
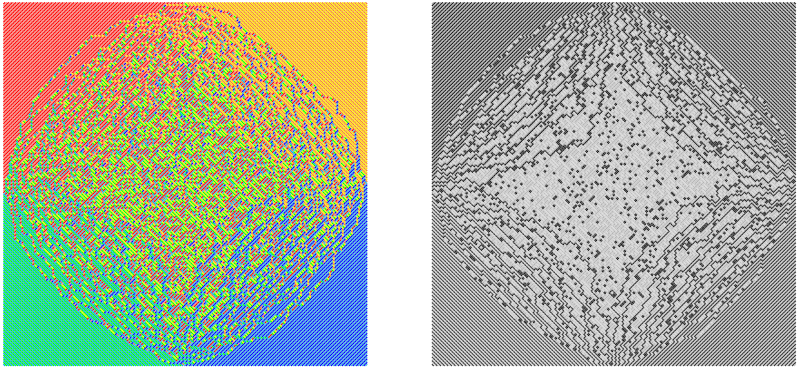

    \centering
    \hfill
    \includegraphics[width=0.4\linewidth]{Figures/ttpn200a05.pdf}
    \hfill
    \includegraphics[width=0.4\linewidth]{Figures/ttpn200a05bw.pdf}
    \hfill\mbox{}
    \caption{Random domino tiling of doubly-periodic Aztec diamond for $b=0.5$
    in two representations, the color representation on the left with eight
  colors to distinguish rough and smooth phases, and with eight different gray
colors on the right.}
    \label{fig:simul_doubly}
  \end{figure}

Based, for example, on computer simulations \hyperref[fig:simul_doubly]{Figure
\ref{fig:simul_doubly}} of domino tilings of the doubly-periodic Aztec diamond
on \hyperref[fig:simul_doubly]{Figure~\ref{fig:simul_doubly}}, we see the presence
of all three phases, frozen, rough and smooth, together with the arctic curve with
two connected components, the outer curve separating frozen and rough regions,
and the inner curve separating rough and smooth phase regions. This is again a
situation where there is a unique point in the liquid region for a given slope
in the Newton polygon.
The restriction of $\psi$ to the annulus~$\mathcal{A}=\mathbb{R}/2\mathbb{Z}\times(0,\Im\tau)$
parametrizes the upper half of the spectral curve $\mathcal{C}^+$, and thus
$\zeta\in\mathcal{A}$ gives a parametrization of the liquid region. The lower boundary of the annulus is mapped to the outer
boundary of the liquid region (the liquid/solid interface): up to a horizontal
translation, we may assume that the turning points correspond to
$\zeta=-\frac{1}{2},0,\frac{1}{2},1$ (mod 2).
The upper
boundary is mapped to the inner boundary of the liquid region (the liquid/gas
interface).

The boundary conditions along the outer boundary are the same as in the uniform
case:
\begin{center}
\begin{tabular}{c | c | c | c | c}
&  (1) & (2) & (3) & (4) \\
& $-\frac{1}{2}<\zeta<0$ & $0<\zeta<\frac{1}{2}$ & $\frac{1}{2}<\zeta <1$ &
$1<\zeta < \frac{3}{2}$\\ \hline
$s$ & $1$ & $0$ & $-1$ & $0$\\
$t$ & $0$ & $-1$ & $0$ & $1$\\
$c$ & $-\frac{1}{2}$ & $\frac{1}{2}$& $-\frac{1}{2}$ & $\frac{1}{2}$
\end{tabular}\end{center}


In order to perform harmonic extensions in a multiply-connected domain, we need
a particular toolbox of special functions. Basically, we are going to apply the
same strategy: first, write down a basic building block, and then construct
harmonic extensions in terms of these building blocks. Since the spectral curve
is a torus, these blocks will be constructed from elliptic functions that we
discuss now.

\subsubsection{Elliptic functions and their classification by Weierstrass functions}

\label{elliptic_appendix}
Let $\Lambda$ be a lattice generated by two vectors $\omega_1$ and $\omega_2$ called periods of $\Lambda$, $\Lambda:=\{n \omega_1+m\omega_2:n,m\in\ZZ \}$. Then, a $\Lambda$-elliptic function is a meromorphic function on $\mathbb{C}$ that satisfies $f(z)=f(z+\omega_1)=f(z+\omega_2)$.

The most famous example of elliptic functions is  Weierstrass's $\wp$-function, which is defined as
\bb
\wp (z, \Lambda):={\frac {1}{z^{2}}}+\sum_{\lambda \in \Lambda \setminus \{0\}}\left({\frac {1}{(z-\lambda )^{2}}}-{\frac {1}{\lambda ^{2}}}\right).
\ee 
The function $\wp$ clearly has poles of order two at each lattice point. It is an elliptic function, which follows from the definition.
For the next subsections, we need two other Weierstrass functions.

The Weierstrass $\zeta$-function is a function fixed by the equation
\bb
 \frac{d\zeta(z,\Lambda)}{dz}=-\wp(z,\Lambda). 
\ee
This defines it up to an additive constant that we fix in a way that $\lim_{z\to 0}\zeta(z,\Lambda)-\frac{1}{z}= 0$.
Using ellipticity of $\wp$ one can derive quasi-periodic properties of $\zeta$ Weierstrass function, $\zeta(z+\omega_i)=\zeta(z)+2\eta_i$, where $\eta_i=\zeta(\omega_i /2)$.

The Weierstrass $\sigma$-function is defined by
\bb
\frac{d \log (\sigma(z,\Lambda)}{dz}=\zeta(z,\Lambda).
\ee
Here, we fix an integration constant so that $\lim_{z\to 0}\sigma(z,\Lambda)-z=0$.
by definition of $\sigma$ and quasi-periodicity of $\zeta$ one derives quasi-periodicity of $\sigma$,
\bb
\sigma(z+\omega_i)=e^{-2\eta_i (z+\omega_i /2)}\sigma(z). \label{elliptic_sigma}
\ee

Function $\sigma$ is a suitable building block of elliptic functions, and every elliptic function can be expressed in terms of Weierstrass $\sigma$ function~\cite{AKH}.
\begin{theorem}
    Suppose $f$ is an elliptic function with periods $\omega_1, \omega_2$ and set of zeroes in a fundamental domain of $\Lambda$ at points $\{\varkappa_i\}$ of the corresponding orders $\{n_i\}$ and set of poles $\{\varkappa_i^{\prime}\}$ of the orders $\{n^{\prime}_i\}$ subject to two conditions:
    \begin{itemize}
      \item $\sum{n_i \varkappa_i}=\sum n_i^{\prime}\varkappa_i^{\prime}$,
        \item $\sum n_i=\sum n^{\prime}_i$.
        \label{elliptic_thm}
    \end{itemize}

And let us define a function $g$ as
\bb
g(z):=\frac{\prod_i \sigma (z-\varkappa_i,\Lambda)^{n_i}}{\prod_i \sigma (z_i-\varkappa^{\prime}_i, \Lambda)^{n_i^{\prime}}}. \label{expression_elliptic}
\ee
Then, the ratio $\frac{f}{g}$ is constant on $\mathbb{C}$.
\end{theorem}
\begin{proof}
First, $g$ is elliptic by quasi-periodicity of $\sigma$ recalled in Equation~\eqref{elliptic_sigma},
and by~\eqref{expression_elliptic}. Simply the translation of the function by a period $\omega_i$ does not change its value by our assumptions, $g(z)\mapsto g(z+\omega_i)=g(z)\exp\left(-2\eta_i (\sum{n_i \varkappa_i}-\sum n_i^{\prime}\varkappa_i^{\prime} )\right)=g(z)$.
Moreover, by~\eqref{expression_elliptic} it has the same zeroes and poles with the same multiplicities as $f$.
Thus, the ratio $\frac{f(z)}{g(z)}$ has no zeroes or poles on $\mathbb{C}$, and thus 
is constant by Liouville's theorem.
\end{proof}

In our case, $\Lambda=2\mathbb{Z}+2\tau\mathbb{Z}$, so that $\omega_1=2$ and
$\omega_2=2\tau$.

Theorem \ref{elliptic_thm} allows us to explicitly write the harmonic extensions of piecewise constant boundary conditions, which we encounter in the tangent plane method. Recall the simplest Dirichlet boundary condition on $\mathbb{H}$, 1 on the negative real line, and 0 on $\RR_{\RR\geq 0}$. The harmonic extension is given by $\frac{1}{\pi}\arg z$.
Now, if we look at annulus $\mathbb{A}_{r1,r2}=\{z\in \CC | r1<|z|<r2\}$ with boundary condition of $0$ on the outer circle and $1$ on the arc connecting $\omega_0$ and $\omega_1$ on the inner circle.
The harmonic extension in this situation is $\frac{1}{\pi} \arg \frac{z-\omega_0}{z-\omega_1} \frac{\log \frac{|z|}{r_2}}{\frac{r_1}{r_2}}$.

Now, suppose we have interval $(a,b)$ on the upper side of side of the rectangle $[0,1]\times[0,|\tau|]$. Then, the harmonic extension of boundary condition $1$ on this interval and $0$ elsewhere is given by the following formula

\bb
h(u)=\frac{1}{\pi}\arg \frac{\sigma(u-b-\tau)}{\sigma(u-a-\tau)}+\frac{2\eta_1(1) (b-a)\Im (\tau-u) }{|\tau|}.
\label{eq:elliptic_harm_ext}
\ee

Take a point $u$ on the upper boundary and the complex conjugate $u^{\star}$. On one hand, $f(u^{\star})=f(u)$ as their difference is exactly the period of elliptic function. But on the other hand, since $f(z)$ takes real values on $\RR$, by the Schwarz reflection principle we have $f(\bar{z})=\overline{(f(z))}$. Thus, the values at points $u, u^{\star}$ are real since they do not change under complex conjugation($f(u)=f(u^\star)=\overline{f}(u)$). Therefore, the argument from \eqref{eq:elliptic_harm_ext} vanishes.

We also need to take derivatieves of such extensions. Let us compute the $\partial_u h(u)$ for which we take the branch of argument by $\arg u=\Im \log u$ for the logarithm with branch cut along the real line.
First, for $z=x+i y$ we have use the Wiertinger derivative to deferentiate with respect to $z$ (as $\Im z$ is not a holomorphic function).

\bb
\frac{\partial}{\partial z} \Im z= \frac{1}{2}\left(\frac{\partial}{\partial x} -i\frac{\partial}{\partial y} \right)y=\frac{1}{2 i}.
\ee
Then, by the chain rule we have
\begin{multline}
\partial_u h(u)=\partial_u\!\left(\Im\log\frac{\sigma(u-b-\tau)}{\sigma(u-a-\tau)}+2\eta_2(\tau)(b-a)\frac{\Im(\tau-u)}{|\tau|}\right)=\\ 
\frac{1}{2i}\left(\zeta(u-b-\tau)-\zeta(u-a-\tau)-2\eta_2(\tau)(b-a)\right).
\end{multline}

\subsubsection{Applying the tangent plane method}

The solution from~\cite{KP:new}, adapted to our normalization, is the following,

\begin{align*}
  s(\zeta)&=\frac{1}{\pi}\arg
  \frac{
    \sigma(\zeta)\sigma(\zeta-\frac{1}{2})
  }{
    \sigma(\zeta-1)\sigma(\zeta+\frac{1}{2})
  },\\
  t(\zeta)&=\frac{1}{\pi}\arg
  \frac{
    \sigma(\zeta+\frac{1}{2})\sigma(\zeta)
  }{
    \sigma(\zeta-\frac{1}{2})\sigma(\zeta+1),
  },\\
  c(\zeta)&=\frac{1}{2}+\frac{1}{\pi}\arg\frac{\sigma(\zeta+\frac{1}{2})\sigma(\zeta-\frac{1}{2})}{\sigma(\zeta-1)\sigma(\zeta)}-\frac{1}{4}\Im(\zeta)
\end{align*}

\begin{figure}[htpb]
  \centering
  \includegraphics[width=0.45\linewidth]{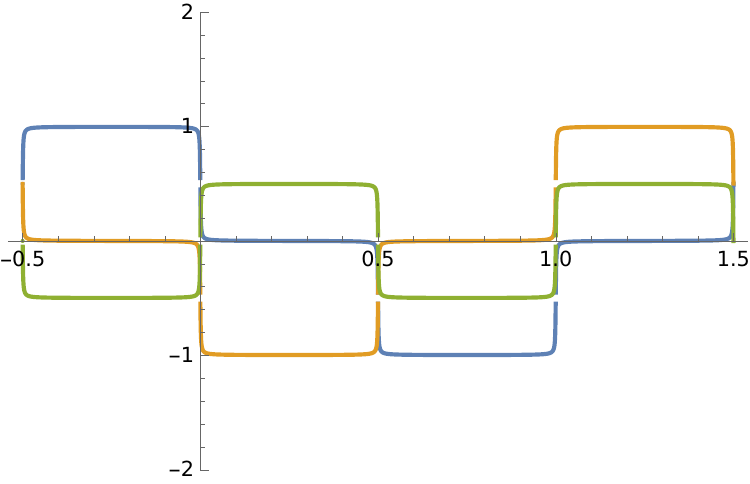}
  \caption{Plot of the functions $s$, $t$, and $c$ for $\zeta$ on (in fact very
    close to) the lower boundary of the cylinder $\mathcal{A}$, matching the
  table, $s$ in blue, $t$ in yellow and $c$ in green.}
  \label{fig:plot_s_t_c_2perio}
\end{figure}
The functions $s$ and $t$ are arguments of genuine elliptic functions in the variable $\zeta$
on the whole torus, by Theorem~\ref{elliptic_thm}. For the function $c$, the
zeros and the poles of the fraction in the argument are fixed by the wanted
behavior along the outer boundary, but this does not define a proper elliptic
function as it is not of the form given by Theorem~\ref{elliptic_thm}. We need
to add the multiple of the imaginary
part of $u$ to guarantee periodicity in the horizontal direction. The
coefficient $\frac{1}{4}$ is specific for the choice of $\tau=1$. For an
arbitrary value of $\tau$, it should be replaced by $\frac{\eta_1}{\pi}$ (when
$\omega_1=2$ and $\omega_2=2\tau$, then $\eta_1=\frac{\pi}{4}$).

Writing the master equation for the harmonically moving planes,
\begin{equation*}
  s_\zeta x + t_\zeta y +c_\zeta = 0,
\end{equation*}
solving for $x=x(\zeta)$ and $y=y(\zeta)$ for each $\zeta$ on the boundary of the
annulus $\mathcal{A}$ gives a parametrization of the arctic curve. This curve is
represented on Figure~\ref{fig:curve_doubly}: the yellow connected component
corresponds to $\Im(\zeta)=0$, whereas the blue component, delimiting the
boundary of the gas bubble, corresponds to $\Im(\zeta)=1$. We can determine the
equation of the tangent plane to the gaz bubble by evaluating $s$, $t$ and $c$
along the blue boundary, and obtain that with our convention, that plane is the
plane $z=0$.

\begin{figure}
    \centering
    \includegraphics[width=0.5\linewidth]{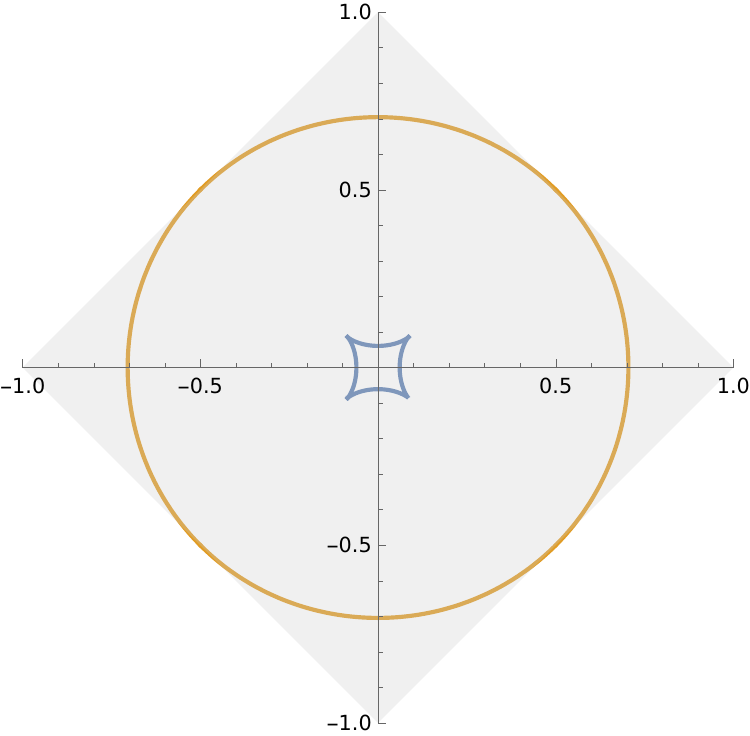}
    \caption{Frozen curve for the doubly-periodic Aztec diamond. The grey square
    is the limiting domain of a renormalized large Aztec diamond. The liquid
    phase, homeomorphic to an annulus, is bounded on the outside from the four
    frozen phases by the yellow component, and bounded in the inside from the
    gas bubble by the blue components. The frozen regions,
    labeled (1), (2), (3), (4) in the table, appear in the
    right, top, left, bottom corner respectively. The equation of the corresponding
    plane can be read from the matching column in the table. The one for the gas
  bubble is $z=0$.}
  \label{fig:curve_doubly}
\end{figure}

In the next section, we will reuse these expressions in a more complicated situation.

\subsection{Formulation of problem for Aztec diamond with a hole}
\label{Our_case}
Recall that Aztec diamond of order $N$, $AD_N$ is the set of unit squares of $\ZZ^2$ whose centers $(x,y)$ satisfy $|x| +|y|\leq N$. And fix $0<\varkappa<1$, then an \textit{Aztec diamond with a hole} of order $N$ is the standard Aztec diamond $AD_N$ with a hole consiting of a smaller centered Aztec diamond of order $\varkappa N$. Another point of view is that we consider a conditional probability by fixing the configuration of the smaller Aztec diamond.

Here, we perform computations for the Aztec diamond with a hole. Furthermore, instead of letter $\zeta$, we denote the intrinsic coordinate by $u$ to distinguish it from $\zeta$-Weierstrass function.
\begin{figure}[h!]
    \centering
    \includegraphics[width=0.4\linewidth]{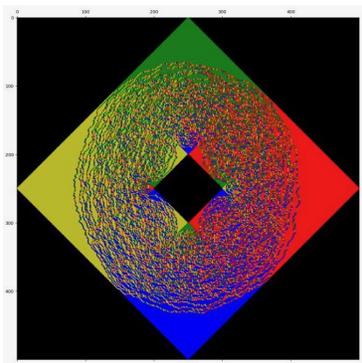}
    \caption{Computer simulations of random domino tiling of Aztec diamond with a hole for $\varkappa=1/4$ and order of the external diamond $N=500.$}
    \label{fig:aztec5}
\end{figure}
First, using computer simulation \hyperref[fig:aztec5]{Figure \ref{fig:aztec5}} we can assume that the number of frozen phases equals $16$, $4$ phases around the external boundary, and $12$ around the internal boundary.
We know that the frozen phases are images of corners of $\mathcal{N}$, that are also labeled by $4$ types of dominoes. Therefore, the degree of $\nabla \mathfrak{h}$ is $4$, so is the degree of $z$. We also see from the computer simulation that the topology of the liquid region is the topology of an annulus. Therefore, we stay in the same frame as in the doubly-periodic Aztec diamond. What changes is the boundary conditions on the upper boundary; now it consists of $12$ non-trivial intervals $(a_i,a_{i+1})$, subject to $a_{i+3}=a_{i}+1/2$, which is a consistency condition to match the upper boundary with the lower one (from the simulation, we see that the frozen phases on both components change simultaneously at four points, so it is an expected condition for the uniform distribution of domino tilings). We still have the freedom to place three points on the first interval out of four. We take them by symmetry $a_1=\tau-1/2-a, a_2=\tau-1/2,a_3=\tau-1/2+a$ for a small parameter $a$ yet to be determined. The lenght of the first and the last sub-interval equal $a$. Therefore, this parameter is resposible for the location of one of the cusps between two frozen regions close to the inner boundary.

The lower boundary has the same boundary conditions as $\mathcal{AD}$, and it consists of $4$ intervals $(a_i,a_{i+1})$, $13 \leq i \leq 16$. Thus, for the lower boundary conditions, we can use the expression for the ring from the section above (it will not change the behavior on the line $\Im u=\Im \tau$ since those extensions are zero there).
\begin{center}
\begin{tabular}{c|c|c|c|c|c|c}
   u  & $(a_1,a_2)$& $(a_2,a_3)$& $(a_3,a_4)$&$(a_4,a_5)$&$(a_5,a_6)$&$(a_6,a_7)$ \\
   s  & 0 & -1& 0&1&0& -1 \\
   t  & 1 &0 & -1&0& 1& 0   \\
   c & +$\varkappa/2$&-$ \varkappa/2$& -$3\varkappa/2$& $-\varkappa/2$& $\varkappa/2$ & $3\varkappa/2$
\end{tabular}
\label{table_aztec}
\end{center}
$\varkappa$ parametrize the size of the hole, $0<\varkappa<1$. The boundary data for $c(u)$ is obtained by analogy with the Aztec diamond \hyperref[table_aztec]{(\ref{table_aztec})}, $c(u)$ changes by $1$ each time we move from one frozen region to the other. Since we removed the Aztec diamond of scale $\varkappa$, we have the same boundary conditions for $c$ up to a global scaling by $\varkappa$.

Functions $s,t$ and $c$ are periodic on the universal cover of $\mathcal{A}$, and can be found as arguments of ratios of $\sigma$ Weierstrass functions. 
The basic block of our harmonic extensions is a periodic step function defined on the infinite strip $\mathcal{S}=\RR\times[0,\tau]$, which is the universal cover of the ring $\mathcal{A}$,
\bb
\frac{1}{\pi}  \arg \left(\frac{\sigma(u-b)}{\sigma(u-a)}\right)
\ee
The plot of this function is as a step function with height $1$ from $a$ to $b$ defined on the real axis,
\begin{figure}[h!]
    \centering
    \includegraphics[width=0.3\linewidth]{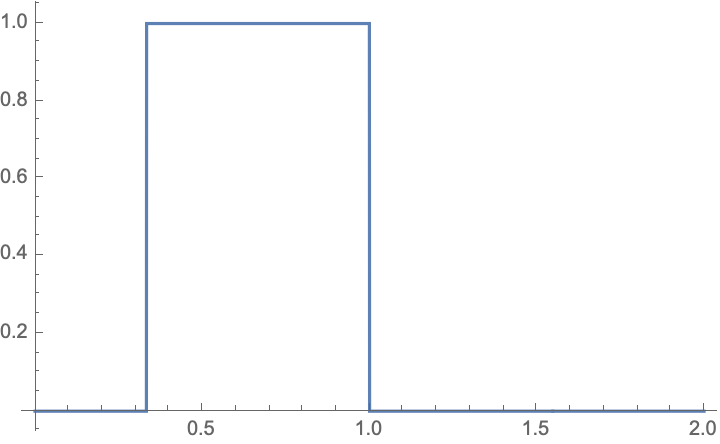}
    \caption{Example of plot of function $S$ for $a=\frac{1}{3}$ and $b=1$.}
\end{figure}


The symmetry of Aztec diamond implies the harmonic extension of the following form,

\begin{multline}
     s(u) = \frac{1}{\pi} \arg  \frac{\sigma(u)\sigma(u-1/2) \sigma(u - (\tau - 1/2)) \sigma(u - (\tau - 1/2 + a)) \sigma(u - (\tau + a)) }{\sigma(u - 1) \sigma(u + 1/2) \sigma(
       u - (\tau - 1/2 - a)) \sigma(u - (\tau - a)) \sigma(
       u - \tau) }+
\\
       +\frac{1}{\pi}\arg\frac{\sigma(u - (\tau + 1 - a)) \sigma(
       u - (\tau + 1))\sigma(u - (\tau + 1/2 - a))}{\sigma(
       u - (\tau + 1/2 + a)) \sigma(u - (\tau + 1 + a))\sigma(u - (\tau + 1/2))}.
       \label{extension_s}
\end{multline}

The function $t(u)$ can be obtained by symmetry from $s(u)$ as $t(u)=s(u+1/2)$, which is, again, expected for the uniform measure. These two functions are arguments of meromorphic functions as the product of $\sigma$ function under the $\arg$ 
satisfy conditions of Theorem \ref{elliptic_thm}. For function $c$, however, the situation slightly different,

\begin{multline}
    c(u)= \frac{1}{\pi}\arg \frac{\sigma(u + 1/2) \sigma(u - 1/2)}{\sigma(u - 1) \sigma(u)} + 
 \frac{\varkappa}{\pi} ( \arg \frac{ \sigma(u - (\tau - a))}{\sigma(u - (\tau - 1/2 - a))} + 
    \arg\frac{\sigma(u - \tau)}{\sigma(u - (\tau - 1/2))} + 
        \\
    \arg\frac{\sigma(u - (\tau + a))}{\sigma(u - (\tau - 1/2 + a))}
    +\arg\frac{\sigma(u - (\tau + 1 - a))}{\sigma(u - (\tau + 1/2 - a))} + 
    \arg\frac{\sigma(u - (\tau + 1))}{\sigma(u - (\tau + 1/2))})
    \\
    +\frac{1}{\pi}\arg \frac{\sigma(u - (\tau + 1 + a))}{\sigma(u - (\tau + 1/2 + a))}+K(a,\varkappa,\tau)-\delta \Im u.
    \label{form_cu}
\end{multline}
Here, the product of $\sigma$ functions does not satisfy Theorem \ref{elliptic_thm} like in the doubly-periodic Aztec diamond. Therefore, the resulting function changes after a shift by a period; we subtract this change with the help of $\zeta$-Weierstrass functions; we call this term $K(a,\varkappa,\tau)$. We have $\frac{1}{2}\Im u$ from the doubly-periodic Aztec diamond term. By the same reason, the function $c$ is not yet elliptic per se. Therefore, we need to subtract a linear factor of $\Im u$. We have $K(a,\varkappa,\tau)=\frac{1}{2}+\frac{\eta_1}{\pi} (3\varkappa+\frac{1}{2})$. We also have a term $-\delta \Im u$, which is responsible for the height change $r$. Parameters $\delta$ and $r$ are not equal to each other, that can be seen on \ref{fig:parameters}. This term, however, violates the ellipticity of the function $c(u)$. That means that the function changes by an additive constant after addition of multiples of $\tau$, yet in the real direction, it is periodic.

\subsubsection{Plot of the height function}

In order to produce a parametric plot of the surface given by $\mathfrak{h}^{\star}$, we need to check that this surface is well-defined. In practice, it means that we need to look at $z$, which is related to conformal coordinate $u$ by a rational transformation of degree $4$(this is because, based on computer simulations, each frozen phase repeats $4$ types going around the boundary).

At the critical points of $z(u)$, derivative $z_u$ vanishes, therefore, the whole expression 
$ s_z z_u  + t_z z_u  + c_u=0$ vanish as well. Thus, we must have $c_u=0$ at those points.
Combining \eqref{extension_s} with \eqref{pre-beltrami} we see that $z(u)$ is given by the expression in the argument of $\arg$ in $s(u)$,

\begin{multline}
z(u) = 
\left( \frac{
\sigma(u + \tfrac{1}{2}) \, \sigma(u)
}{
\sigma(u - \tfrac{1}{2}) \, \sigma(u + 1)
} \right)
\cdot
\frac{
\sigma(u - \tau + 1) \, \sigma(u - \tau - a + 1)
}{
\sigma(u - \tau + a + 1) \, \sigma(u - \tau + a + \tfrac{3}{2})
} \\
\cdot
\frac{
\sigma(u - \tau - a + \tfrac{1}{2}) \,
\sigma(u - \tau + a) \,
\sigma(u - \tau + a + \tfrac{1}{2}) \,
\sigma(u - \tau + \tfrac{3}{2})
}{
\sigma(u - \tau + \tfrac{1}{2}) \,
\sigma(u - \tau + 1) \,
\sigma(u - \tau - a + 1) \,
\sigma(u - \tau - a + \tfrac{3}{2})
}
\end{multline}


Also, there is a similar expression for $w(u)$, which we omit. However, its plot shows that there are two critical point inside the liquid region in the fundamental domain.

\begin{figure}
    \centering
    \includegraphics[width=0.5\linewidth]{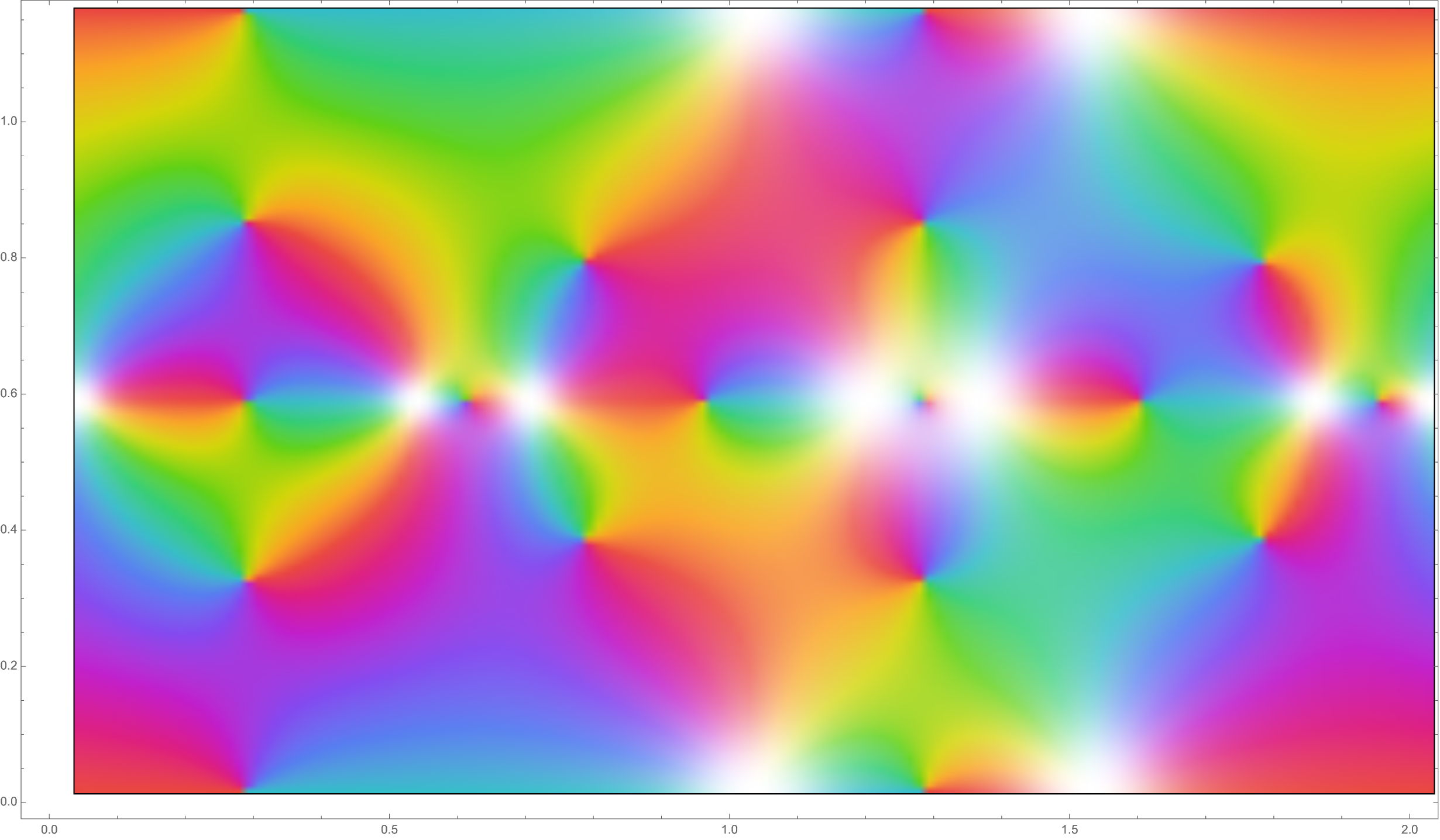}
    \caption{Plot of the derivative of $w(u)$ on the whole torus with two critical points inside the liquid region repeated $4$ times}
    \label{fig:crit_w}
\end{figure} 

On the figure \ref{fig:crit_w} we see analytic properties of $w(u)$ on the whole torus. There are $8$ zeros and $8$ poles on the boundary of $\pp \mathcal{L}$, which are real zeros that we see directly from the expression. Moreover, there are $2$ complex zeroes repeated $4$ times in the interior of $\mathcal{L}$. They are linked by two involutions, the complex conjugation, and symmetry $w(u+1)=w(u^{-1})$.

\subsection{Determine the constants}
\label{scheme}
So far we have constants $a,\varkappa,\tau$ from harmonic extensions, yet there's an extra parameter encoding the height change $r$. That is let $\delta\geq 0$. We substract from $c(u) $ term $\delta\Im(u)$ that does not change the value on the lower boundary, but does on the upper one. We do not know the actual height change, and therefore, a priori can not assume any particular value of $\delta$.

For a given value of $\tau$ and $\delta$, we find $a^{\tau}$ such that the two critical points of $z(u)$ have the same imaginary part (the real parts are $0.25$ and $0.75$). A good way of doing it, it using twice the \textit{FindRoot} function in Mathematica. Then, we use $a^{\tau}$ to find $\varkappa^{\delta,\tau}$ such that the two critical points of $cu(u,a^{\tau},\delta,\varkappa^{\delta,\tau})$ have the same imaginary parts. Then, we end up with critical points of $cu(u,a^{\tau},\delta,\varkappa^{\delta,\tau})$ and $z$ lying on two horizontal lines, and we tune $\tau$ numerically using the bisection method to ensure that the two lines coincide at $\tau_{\delta}$ with parameters $a^{\delta,\tau_{\delta}}, \varkappa^{\delta,\tau_\delta}$ (we recompute each parameter on every step of the scheme). 
We also checked that for $a^{\delta,\tau_{\delta}}$ and $\varkappa^{\delta,\tau_{\delta}}$ we have unique roots, see the plot of the difference between the imaginary parts of the roots of $c_u$ as a function of $\varkappa$ on fig \ref{fig:unique_plot}.
\begin{figure}
    \centering
    \includegraphics[width=0.5\linewidth]{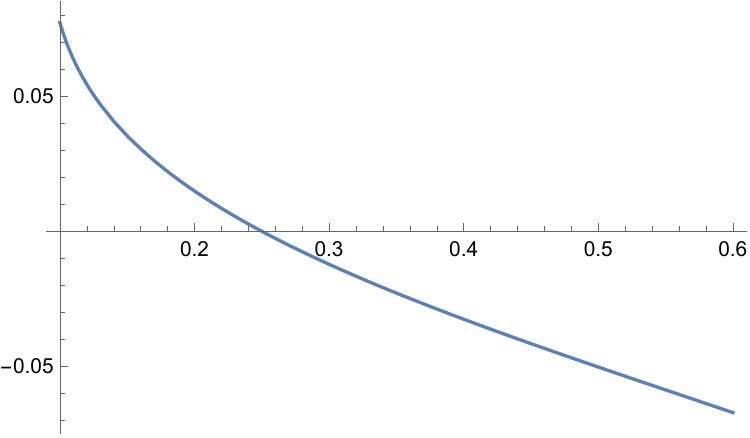}
    \caption{the plot of the difference between the imaginary parts of the roots of $c_u$ as a function of $\varkappa$. Clearly, there is a unique root.}
    \label{fig:unique_plot}
\end{figure}

In principle, we have one parameter $\delta$, which, however, is not a "physical parameter" of the model. Such a parameter is the size of the inner Aztec diamond $\varkappa$, and one could numerically parametrize everything by it. 

Moreover, it turns out that the imaginary part of $cu(u^{\star})$ at the critical points $u^{\star}$ vanishes, which allows us to express $\varkappa$.
Recall the structure of function $cu(u,\varkappa)$ from \eqref{form_cu}
\bb
cu(u^{\star})=cu_1(u^{\star})+\varkappa cu_2(u^{\star})
\ee
Therefore, chosing
\bb
\varkappa(\tau)=-\frac{cu_1(u^{\star})}{cu_2(u^{\star})} \label{elliptic_size}
\ee
we get the desired coalision of the critical points.

See plot with numeric dependence of parameters on $\varkappa$ fig \ref{fig:parameters}.

Now, we are able to produce plots of the height function for different values of $\delta$. By \cite{KO:burgers}, the height change can be identified with a coordinate on the moduli spaces of nodal curves. And by our Theorem 2 from \cite{Kuchumov:dominoes}, if we do not put a condition on the height change, there is a typical value $r(\delta^\star)$ for the uniform distribution on domino tilings.
\begin{figure}[H]
    \centering
    \begin{subfigure}[t]{0.45\textwidth}
        \centering
        \includegraphics[width=\linewidth]{Figures/Curve_delta=0.pdf}
        \caption{Frozen curve for $\delta=0$, parameters $(a,\varkappa,\tau,\delta)$ on the top}
    \end{subfigure}
    \hfill
    \begin{subfigure}[t]{0.45\textwidth}
        \centering
        \includegraphics[width=\linewidth]{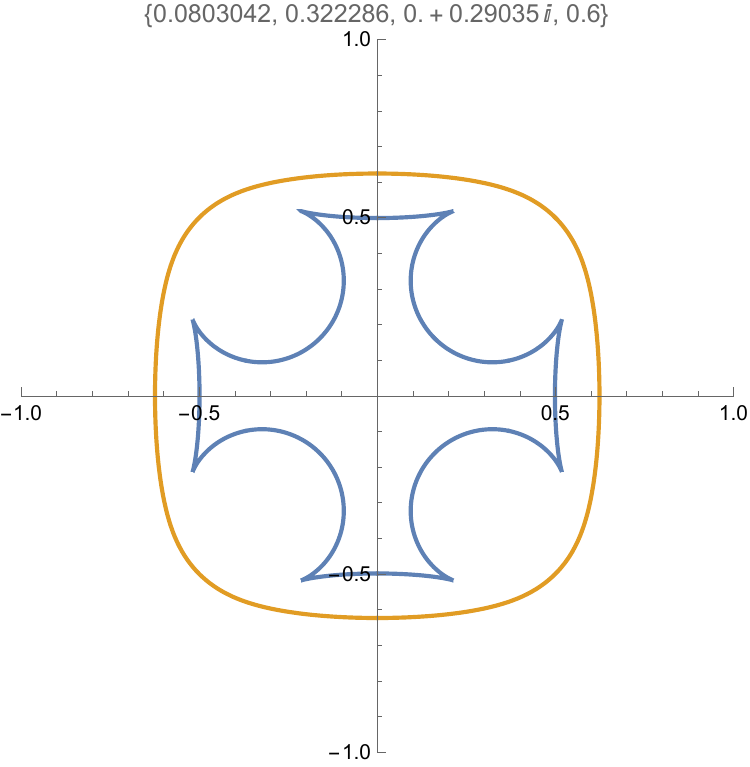}
        \caption{Frozen curve for $\delta=0.60$}
    \end{subfigure}

    \vspace{0.5cm}

    \begin{subfigure}[t]{0.45\textwidth}
        \centering
        \includegraphics[width=\linewidth]{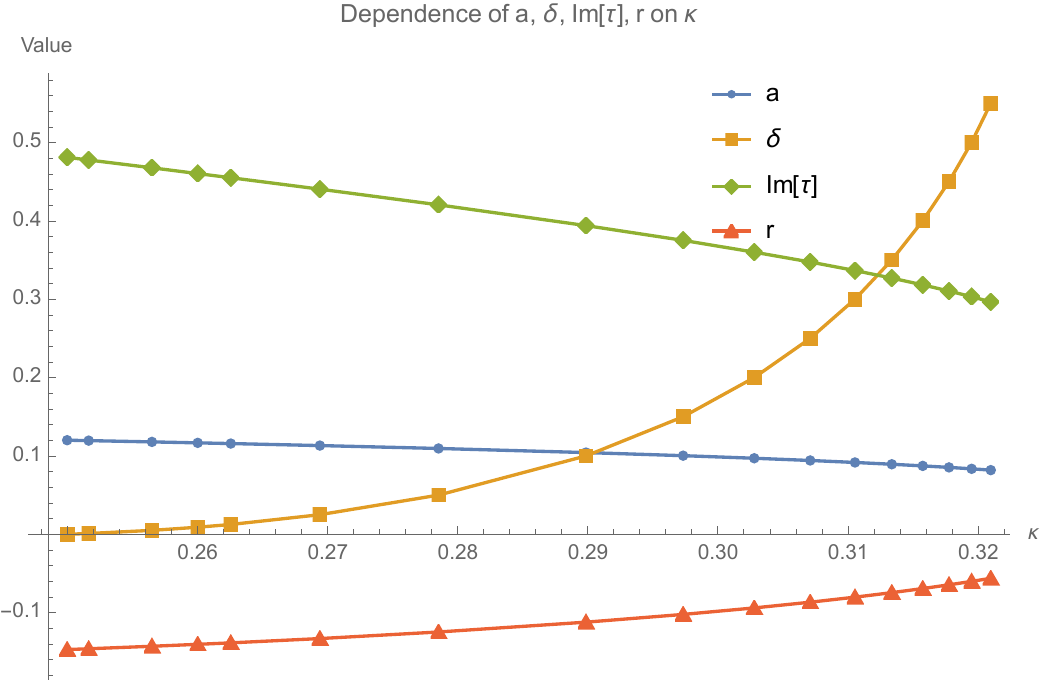}
        \caption{Dependence of parameters on the size $\kappa=\varkappa$ of inner Aztec diamond. Number $r$ is the height change computed explicitply between points $u=0.125+0.001\tau$ and $u=0.125+0.999\tau$}
        \label{fig:parameters}
    \end{subfigure}
    \hfill
    \begin{subfigure}[t]{0.45\textwidth}
        \centering
        \includegraphics[width=\linewidth]{Figures/delta=0.15_3d.pdf}
        \caption{Plot of the limiting height function for $\delta=0.15$}
    \end{subfigure}
\end{figure}

\subsubsection{Arctic curves for the generic parameters}

One can also produce plots of the frozen curve even without fitting the rest set of parameters $(a,\varkappa,\tau,\delta)$. Those computations do not correspond to an actual limit shape as the condition on critical points is not satisfied.
For intermediate values of $\delta$ we obtain a deformed picture as on \hyperref[fig:table_plot]{Figure \ref{fig:table_plot}}.

\begin{figure}
    \centering
    \includegraphics{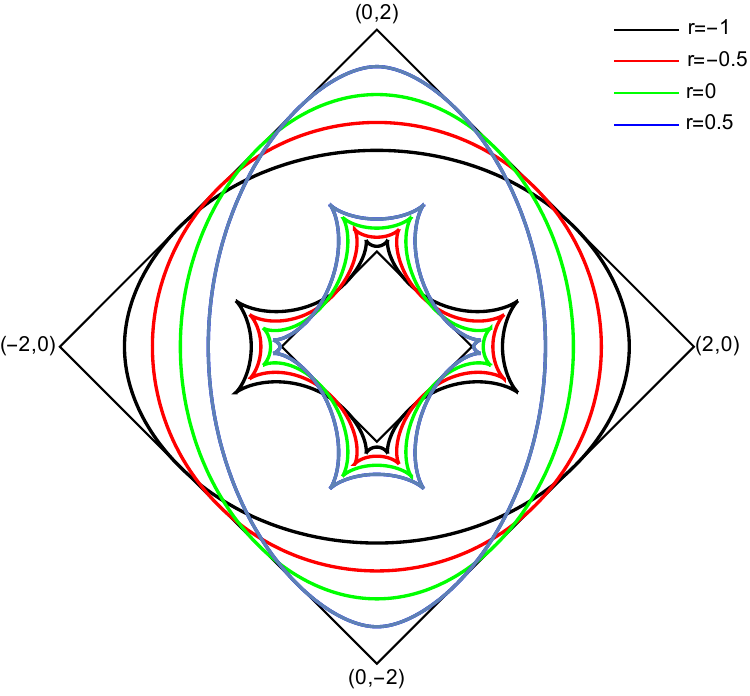}
    \caption{Limit shape for the Aztec diamond with a hole for different values of parameter $\delta$ and $\varkappa=0.3$.}
    \label{fig:table_plot}
\end{figure}

We can also perform our computations for various sizes of the hole $\varkappa$.

\begin{figure}
    \centering
    \begin{subfigure}{0.4\linewidth}
    \includegraphics[width=\linewidth]{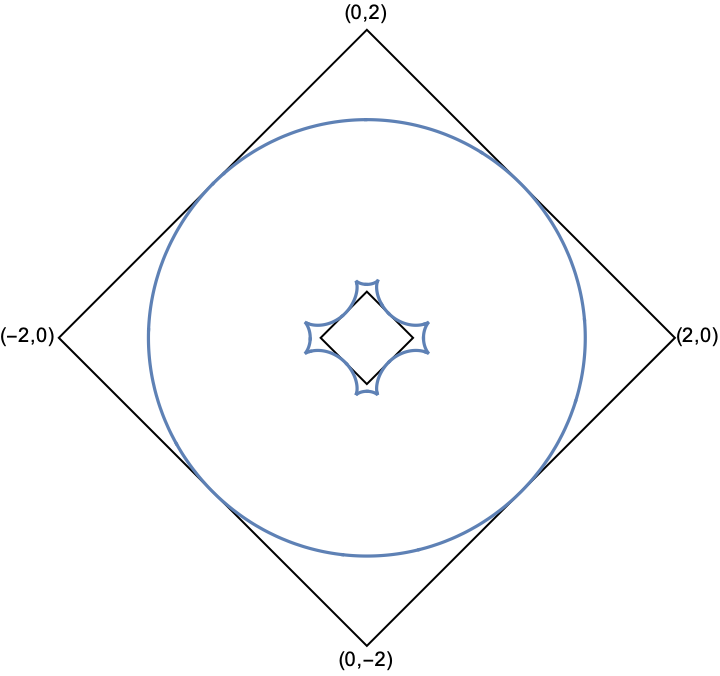}
    \end{subfigure}
~
\begin{subfigure}{0.4\linewidth}
        \includegraphics[width=\linewidth]{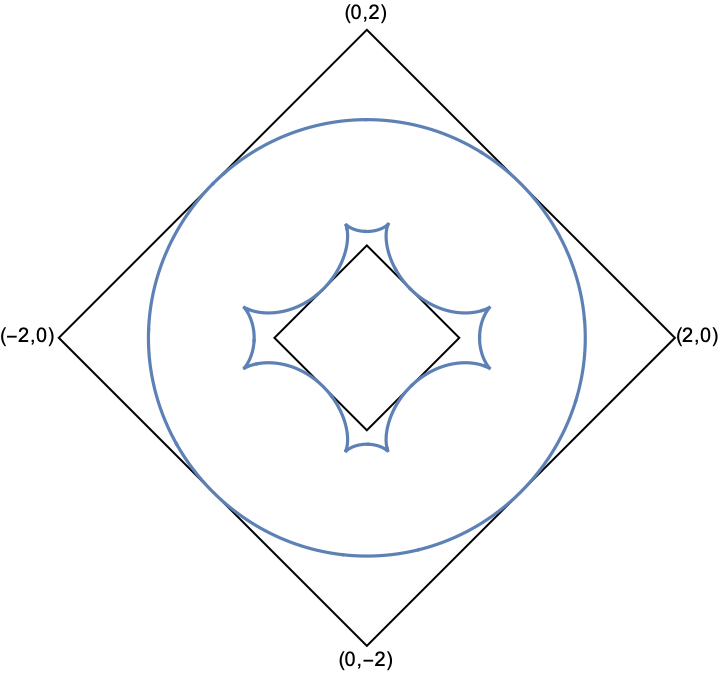}
\end{subfigure}
    \caption{Plot of the frozen curve with $\varkappa=0.15$ on the left, and for $\varkappa=0.6$ on the right.}
    \label{fig:kappa_1_5}
\end{figure}

\begin{figure}
    \centering
    \begin{subfigure}{0.4\linewidth}
    \includegraphics[width=\linewidth]{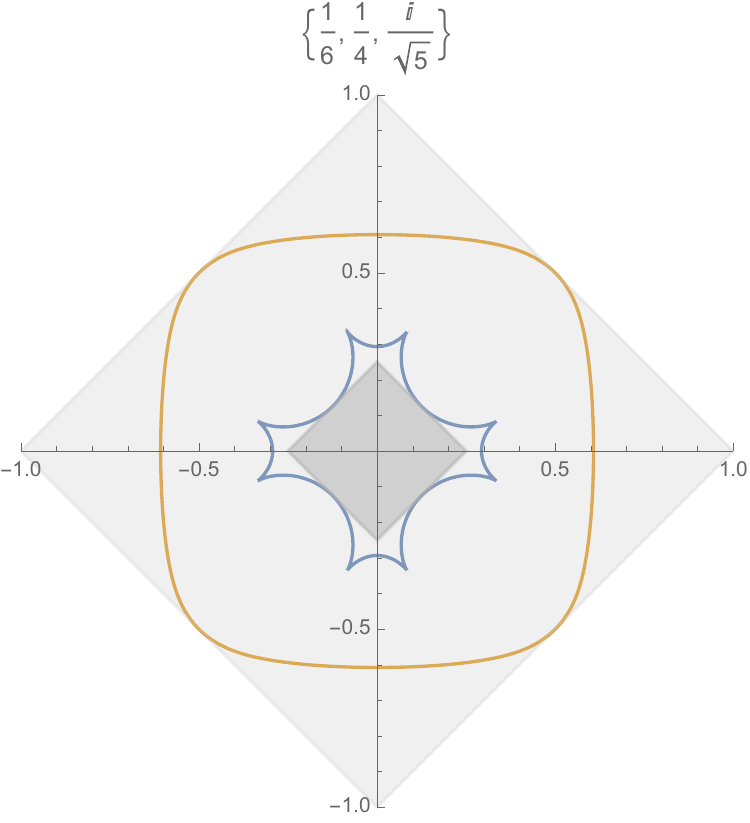}
    \end{subfigure}
~
\begin{subfigure}{0.4\linewidth}
        \includegraphics[width=\linewidth]{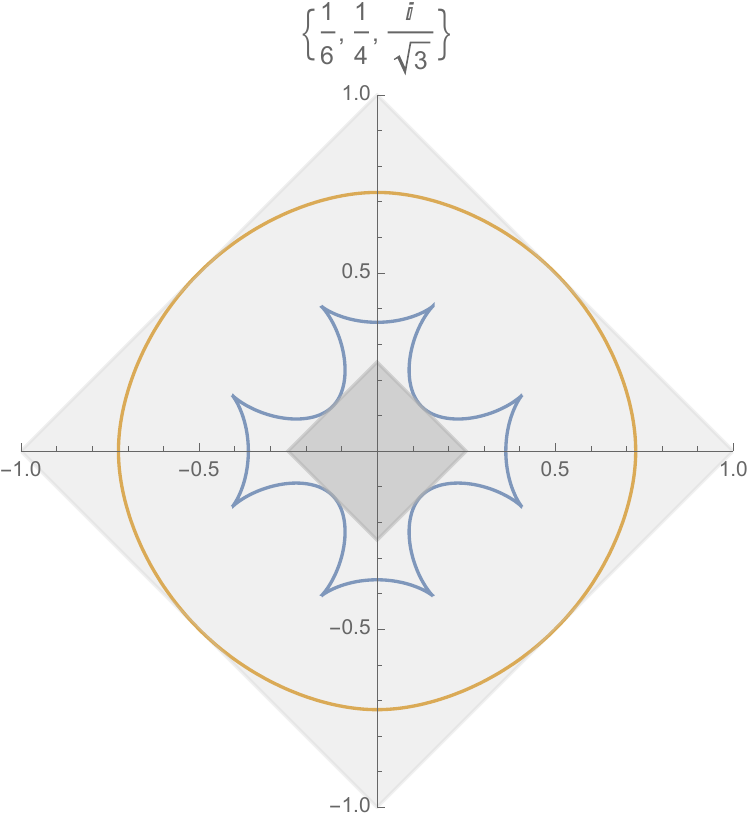}
\end{subfigure}
    \caption{Plot of the frozen curve with $\tau=\frac{\sqrt{-1}}{\sqrt{5}}$ on the left, and for $\tau=\frac{\sqrt{-1}}{\sqrt{3}}$ on the right.}
    \label{fig:tau_deform}
\end{figure}




\appendix
\section{Conformal coordinates}
\label{appendix:beltrami}
Here we discuss conformal coordinates and the Beltrami equation on the Newton polygon $\mathcal{N}$ with a surface tension $\sigma$ defined on it.

Since $\sigma$ is convex, the Hessian  of $\sigma$ $H_\sigma$ defines a non-degenerate metric $g$ on $N\setminus\mathscr{G}$,

\bb 
g=\sigma_{ss} ds^2+2\sigma_{st} dsdt +\sigma_{tt} dt^2.
   \label{metric_sigma}
\ee

Also note that, by Theorem 5.5 in \cite{KOS}, we have \(\det H_\sigma = \pi^2\).
 
In the conformal coordinate $u:=U(s,t)+i V(s,t)$ the metric $g$ writes as
\bb
g=\rho(U,V)^2(dU^2+dV^2)
\label{conformal}
\ee
for a function $\rho(U,V)$.
conformal coordinates exist on two-dimensional manifolds for metric with Hölder coefficients by result from \cite{korn,Lichtenstein}. For example, Theorem 18 i \cite{Spivak} is the following, 
\begin{theorem}
    Let $\Sigma$ be a $C^{\infty}$-smooth two-dimensional surface embedded into $\RR^2$ with metric $g=\langle,\rangle$, and let $p\in\Sigma$ be a point, whose components with respect to the standard coordinate system are real analytic. Then, there exists a real analytic conformal coordinate system in a neighborhood of $p$.
    \label{Thm:Belt}
\end{theorem}

Metric $g=\begin{pmatrix} a & b \\ b & c \end{pmatrix}$
defines scalar product $\langle \cdot, \cdot \rangle_g$ on vector fields on $\mathcal{N}$, and we can determine the conformal coordinates after comparing $\langle \frac{\partial}{\partial U}, \frac{\partial}{\partial V} \rangle_g$ in coordinates $s,t$ and $U,V$. In latter coordinates, it equals zero by definition, while for coordinates $s,t$ it becomes an equation. This condition together with $\langle \frac{\partial}{\partial V}, \frac{\partial}{\partial V}\rangle_g=
\langle \frac{\partial}{\partial U},\frac{\partial}{\partial U} \rangle_g$ results in the system of two real equations, which we call the Real Beltrami equations.
\begin{align}
  U_s=\rho (b V_s-aV_t),\\
  U_t=-\rho (bV_t-cV_s).  
\end{align}
This system naturally arises in the context of the dimer model as we are going to see in this Chapter.
Further, it is equivalent to one complex equation for a complex coordinate $u$ on $\mathcal{N}$, in terms of coordinates $u=U+iV, \bar{u}=U-iV$, the real Beltrami equations result in the usual Beltrami equation, used for example in \cite{KP:plane}.
In terms of $u_s,u_t$, it is the following equation,
\bb
\frac{\bar{\partial} u}{\partial u}=\frac{ \frac{1}{2} (u_s+i u_t)}{\frac{1}{2}(u_s-i u_t)}=\frac{1}{\overline{\mu_\sigma}}.
\label{Beltrami}
\ee
where the Beltrami coefficient $\mu_\sigma$ is given by
\bb
\mu:=\frac{\sigma_{ss}-\sigma_{tt}+2 i \sigma_{st}}{\sigma_{ss}+\sigma_{tt}-2\sqrt{\sigma_{ss}\sigma_{tt} -\sigma_{st}^2}}.
\ee
Also, $\rho(U,V)=\sqrt{\sigma_{ss}+\sigma_{tt}-2\sqrt{\sigma_{ss}\sigma_{tt}}}$.

Existence of a solution for this equation follows from the strict convexity of $\sigma$ on $\mathring{\mathcal{N}}$, or in terms of Beltrami coefficient $|\mu_\sigma|<1$. 
It is worth noting that the Beltrami equation for surface tension corresponding to the dimer model on hexagonal lattice is equivalent to the complex Burgers equation studied in \cite{KO:burgers}.

\begin{proof}[Proof of Theorem \ref{Thm:Belt}]
    
Let us first find the Beltrami equation for real coordinates $(U,V)$, let us start with $\langle \frac{\partial}{\partial U}, \frac{\partial}{\partial V} \rangle_g=0$, which is just the definition of conformal coordinates.
In the $(s,t)$ coordinates, it is a non-trivial identity. First, the vector fields in $(s,t)$ coordinates are
\begin{equation*}
   \frac{\partial }{\partial U}=s_U\frac{\partial}{\partial s}+t_U\frac{\partial}{\partial t},
\end{equation*}
\begin{equation}
    \frac{\partial }{\partial V}=\frac{\partial s}{\partial V}\frac{\partial}{\partial s}+\frac{\partial t}{\partial V}\frac{\partial}{\partial t}.
\end{equation}

Their scalar product is

\bb
\sigma_{ss} s_U s_V+\sigma_{st}(s_V t_U+s_U t_V)+ \sigma_{tt} t_U t_V=0.
\label{scalar_prod}
\ee
Let us express \ref{scalar_prod} in terms of derivatives of $U,V$ using the Jacobian matrix,
\begin{equation}
\mathcal{J}=
\begin{pmatrix}
\frac{\partial s}{\partial U} & \frac{\partial s}{\partial V} \\
\frac{\partial t}{\partial U} & \frac{\partial t}{\partial V}.
\label{jacobian}
\end{pmatrix}
=
\begin{pmatrix}
    \frac{\partial U}{\partial s} & \frac{\partial U}{\partial t} \\
\frac{\partial V}{\partial s} & \frac{\partial V}{\partial t}
\end{pmatrix} ^{-1}
\end{equation}
Using $\mathcal{J}$, we express the desired derivatives as

\bb
s_u=V_t/\det \mathcal{J},\\
s_v= -U_t /\det \mathcal{J},\\
t_u=-V_s/\det \mathcal{J},\\
t_v=U_s/\det \mathcal{J}.
\ee

After substitution of derivatives to \ref{scalar_prod} and multiplying by $\det \mathcal{J}$ we obtain
\bb
U_t (V_s \sigma_{st} - V_t \sigma_{ss})+U_s(V_t \sigma_{st}-V_s \sigma_{tt})=0.
\ee
From this equation, we see that $U_s$ and $U_t$ are proportional and thus there is such a $\rho$ that

\begin{equation*}
    U_s=\rho (V_s \sigma_{st} -V_t \sigma_{ss}),\\
\end{equation*}
\begin{equation}
    U_t= -\rho (V_t\sigma_{st} - V_s \sigma_{tt}).
    \label{belt}
\end{equation}
We also have the equation for the diagonal scalar products,
\bb
\langle \frac{\partial }{\partial U},  \frac{\partial }{\partial U} \rangle_g =\langle \frac{\partial }{\partial V},  \frac{\partial }{\partial V} \rangle_g.
\ee
Or more precisely,

\begin{multline}
    \langle \frac{\partial }{\partial U},  \frac{\partial }{\partial U} \rangle_g= s_U^2 \sigma_{ss}+t_U^2\sigma_{tt}+2\sigma_{st}s_U t_U,\\
    \langle \frac{\partial }{\partial V},  \frac{\partial }{\partial V} \rangle_g= s_V^2 \sigma_{ss}+t_V^2\sigma_{tt}+2\sigma_{st}(s_V t_V). 
    \label{scalar_uu}
\end{multline}
Let us substitute (\ref{belt}) into (\ref{scalar_uu}) to find $\rho$. We need first to express the derivatives of $s$ and $t$ in terms of derivatives of $U$ and $V$, and then use equations (\ref{belt}).

\begin{multline}
    \rho=\langle \frac{\partial}{\partial U}, \frac{\partial}{\partial U} \rangle=\frac{1}{\det \mathcal{J}^2}(V_t^2a+2b(-V_s V_t)+cV_s^2),\\
\rho=\langle \frac{\partial}{\partial V}, \frac{\partial}{\partial V} \rangle=\frac{1}{\det \mathcal{J}^2} (aU_t^2+c U_s^2+2b (-U_s U_t)) 
\end{multline}

\begin{multline}
    V_t^2 a-2b V_s V_t+cV_s^2,\\   
    a \rho^2 (b V_t-c V_s)^2+c\rho^2(bV_s-aV_t)^2+2b\rho^2(b V_s-aV_t)(bV_t-cV_s).
\end{multline}
We see from the comparison of coefficients in front of $V_t^2$ in both equations
\begin{multline}
a=\rho^2(a b^2+ca^2-2bab)\\
1=\rho^2(ca -b^2).
\end{multline}
Similarly, for the coefficient in front of $V_s^2$,
\begin{multline}
    c=\rho^2(ac^2+cb^2-2bbc)\\
    1=\rho^2(ac-b^2).
\end{multline}
Finally, the coefficient in front of $V_t \cdot V_s$,

\begin{multline}
    -2b=-2a\rho^2bc-2c\rho^2ba+2b\rho^2(b^2+ac)\\
    1=\rho^2ca+\rho^2pca-\rho^2(b^2+ca)\\
    1=\rho^2(ca-b^2).
\end{multline}

Thus, we deduce that $\rho=\sqrt{H_\sigma}:=\sqrt{\sigma_{ss}\sigma_{tt}-\sigma_{st}^2}$.

In order to obtain the complex Beltrami equation, assume that we have a solution $U,V$ of the real Beltrami equation \eqref{belt}.
The let us look at their complex combination $u:=U+iV$ ($\bar{u}:=U-iV$), and the Wirtinger derivatives with respect to $u, \bar{u}$,

Using the Wirtinger derivatives, we have the following rules of computation of derivatives of a function $w$:

\bb
w_u=\frac{1}{2}(w_U-iw_V),
~
w_{\bar{u}}=\frac{1}{2}(w_U+iw_V)
\ee
and
\bb
w_U=w_u+w_{\bar{u}}, ~
w_V=\frac{w_{\bar{u}}-w_u}{i}.
\ee

Now, suppose $U,V$ satisfy the Beltrami equations, then write 
\bb
2w_{\bar{u}}\sqrt{ac-b^2}=(b-ia+i\sqrt{ac-b^2})V_x+(c-ib-\sqrt{ac-b^2})V_y
\ee
and
\bb
2w_z\sqrt{ac-b^2}=(b+ia+i\sqrt{ac-b^2})V_x+(c+ib+\sqrt{ac-b^2})V_y
\ee

Then, dividing the two equations and calculations, which we omit, we get an equivalent Beltrami equation in holomorphic coordinates,
\bb
\frac{w_{\bar{z}}}{w_z}=\frac{c-a-2ib}{c+a+2\sqrt{ac-b^2}}.\label{Beltrami1}
\ee
Or in other words,
\bb
w_{\bar{z}}=\mu w_z,
\ee where $\mu=\frac{c-a-2ib}{c+a+2\sqrt{ac-b^2}}$ is the Beltrami coefficient.

Thus, we derived the Belrami equation for the conformal coordinates in two forms. Then, by the theory of elliptic PDE \cite{kari}, Beltrami equation admits a solution as long as $|\mu|<1$, which is the condition of existence of conformal coordinates. In the dimer model, it follows from strict convexity of $\sigma$\cite{ADPZ}. Therefore, inside the liquid region, where the gradient $\nabla\mathfrak{h}$ is in the interior of the Newton polygon, we have the existence of conformal coordinates $z$.

\end{proof}

\bibliographystyle{alpha}
\bibliography{biblio}

@article {KOS,
    AUTHOR = {Kenyon, Richard and Okounkov, Andrei and Sheffield, Scott},
     TITLE = {Dimers and amoebae},
   JOURNAL = {Ann. of Math. (2)},
  FJOURNAL = {Annals of Mathematics. Second Series},
    VOLUME = {163},
      YEAR = {2006},
    NUMBER = {3},
     PAGES = {1019--1056},
      ISSN = {0003-486X},
   MRCLASS = {60D05 (82B26 82B41)},
  MRNUMBER = {2215138},
MRREVIEWER = {Michael Pr\"{a}hofer},
       DOI = {10.4007/annals.2006.163.1019},
       URL = {https://doi.org/10.4007/annals.2006.163.1019},
}

@article{CKP,
author = {Cohn, Henry and Kenyon, Richard and Propp, James},
year = {2000},
month = {08},
pages = {},
title = {A variational principle for domino tilings},
volume = {14},
journal = {Journal of the American Mathematical Society},
doi = {10.1090/S0894-0347-00-00355-6}
}

@article{T,
author = {Thurston, William},
year = {1990},
month = {10},
pages = {},
title = {Conway's Tiling Groups},
volume = {97},
journal = {American Mathematical Monthly},
doi = {10.2307/2324578}
}

@article{KO:burgers,
author = {Kenyon, Richard and Okounkov, Andrei},
year = {2005},
month = {08},
pages = {},
title = {Limit shapes and the complex Burgers equation},
volume = {199},
journal = {Acta Math.},
doi = {10.1007/s11511-007-0021-0}
}

@article{ADPZ,
author = {Astala, Kari and Duse, Erik and Prause, István and Zhong, Xiao},
title = {Dimer models and conformal structures},
journal = {Communications on Pure and Applied Mathematics},
year = {2026},
volume = {79},
number = {2},
pages = {340-446},
doi = {https://doi.org/10.1002/cpa.70014},
url = {https://onlinelibrary.wiley.com/doi/abs/10.1002/cpa.70014},
eprint = {https://onlinelibrary.wiley.com/doi/pdf/10.1002/cpa.70014},
}

@article{D2,
author = {Francesco, Philippe and Guitter, Emmanuel},
year = {2018},
month = {03},
pages = {},
title = {Arctic curves for paths with arbitrary starting points: A tangent method approach},
volume = {51},
journal = {Journal of Physics A: Mathematical and Theoretical},
doi = {10.1088/1751-8121/aad028}
}

@book {H,
    AUTHOR = {Hatcher, Allen},
     TITLE = {Algebraic topology},
 PUBLISHER = {Cambridge University Press, Cambridge},
      YEAR = {2002},
     PAGES = {xii+544},
      ISBN = {0-521-79160-X; 0-521-79540-0},
   MRCLASS = {55-01 (55-00)},
  MRNUMBER = {1867354},
MRREVIEWER = {Donald W. Kahn},
}

@incollection {F,
    AUTHOR = {Fournier, J. C.},
     TITLE = {Pavage des figures planes sans trous par des dominos:
              fondement graphique de l'algorithme de {T}hurston,
              parall\'{e}lisation, unicit\'{e} et d\'{e}composition},
      NOTE = {Selected papers from the ``GASCOM '94'' (Talence, 1994) and
              the ``Polyominoes and Tilings'' (Toulouse, 1994) Workshops},
   JOURNAL = {Theoret. Comput. Sci.},
  FJOURNAL = {Theoretical Computer Science},
    VOLUME = {159},
      YEAR = {1996},
    NUMBER = {1},
     PAGES = {105--128},
      ISSN = {0304-3975},
   MRCLASS = {52C20 (05B45 68R05)},
  MRNUMBER = {1398693},
MRREVIEWER = {Richard Kenyon},
       DOI = {10.1016/0304-3975(95)00204-9},
       URL = {https://doi.org/10.1016/0304-3975(95)00204-9},
}

@article {BG,
    AUTHOR = {Bufetov, Alexey and Gorin, Vadim},
     TITLE = {Fourier transform on high-dimensional unitary groups with
              applications to random tilings},
   JOURNAL = {Duke Math. J.},
  FJOURNAL = {Duke Mathematical Journal},
    VOLUME = {168},
      YEAR = {2019},
    NUMBER = {13},
     PAGES = {2559--2649},
      ISSN = {0012-7094},
   MRCLASS = {60B15 (22E65 60K35)},
  MRNUMBER = {4007600},
       DOI = {10.1215/00127094-2019-0023},
       URL = {https://doi.org/10.1215/00127094-2019-0023},
}

@article{Kuchumov:dominoes,
author = {Kuchumov, Nikolai},
year = {2021},
month = {10},
pages = {32},
title = {A variational principle for domino tilings of multiply-connected domains},
journal = {preprint arXiv:2110.06896}
}

@article{JPS,
author = {Jockusch, William and Propp, James and Shor, Peter},
year = {1998},
month = {02},
pages = {},
title = {Random Domino Tilings and the Arctic Circle Theorem}
}

@article{BK,
author = {Bufetov, Alexey and Knizel, Alisa},
year = {2016},
month = {04},
pages = {},
title = {Asymptotics of random domino tilings of rectangular Aztec diamonds},
volume = {54},
journal = {Annales de l'institut Henri Poincare (B) Probability and Statistics},
doi = {10.1214/17-AIHP838}
}

@article{Colomo-Sportiello,
author = {Colomo, Filippo and Sportiello, Andrea},
year = {2016},
month = {09},
pages = {},
title = {Arctic Curves of the Six-Vertex Model on Generic Domains: The Tangent Method},
volume = {164},
journal = {Journal of Statistical Physics},
doi = {10.1007/s10955-016-1590-0}
}

@article {BF,
    AUTHOR = {Borodin, Alexei and Ferrari, Patrik L.},
     TITLE = {Large time asymptotics of growth models on space-like paths.
              {I}. {P}ush{ASEP}},
   JOURNAL = {Electron. J. Probab.},
  FJOURNAL = {Electronic Journal of Probability},
    VOLUME = {13},
      YEAR = {2008},
     PAGES = {no. 50, 1380--1418},
      ISSN = {1083-6489},
   MRCLASS = {82C22 (60K35)},
  MRNUMBER = {2438811},
MRREVIEWER = {N.\ N.\ Ganikhodjaev},
       DOI = {10.1214/EJP.v13-541},
       URL = {https://doi.org/10.1214/EJP.v13-541},
}

@article {J,
    AUTHOR = {Johansson, Kurt},
     TITLE = {The arctic circle boundary and the {A}iry process},
   JOURNAL = {Ann. Probab.},
  FJOURNAL = {The Annals of Probability},
    VOLUME = {33},
      YEAR = {2005},
    NUMBER = {1},
     PAGES = {1--30},
      ISSN = {0091-1798,2168-894X},
   MRCLASS = {60K35 (15A52 33C90 52C20 82B20)},
  MRNUMBER = {2118857},
MRREVIEWER = {Thomas\ Polaski},
       DOI = {10.1214/009117904000000937},
       URL = {https://doi.org/10.1214/009117904000000937},
}

@article{Di_francesco_guitter_aztec,
author = {Francesco, Philippe and Lapa, Matthew},
year = {2017},
month = {11},
pages = {},
title = {Arctic Curves in path models from the Tangent Method},
volume = {51},
journal = {Journal of Physics A: Mathematical and Theoretical},
doi = {10.1088/1751-8121/aab3c0}
}

@article{KP:plane,
author = {Kenyon, Richard and  Prause, István},
year = {2022},
month = {10},
pages = {3003-3022},
title = {Gradient variational problems in $\mathbb{R}^2$},
volume = {16},
journal = {Duke Mathematical Journal}
}

@article{KP:new,
author = {Kenyon, Richard and Prause, István},
year = {2024},
month = {01},
pages = {},
title = {Limit shapes from harmonicity: dominos and the ﬁve vertex model},
volume = {57},
journal = {Journal of Physics A: Mathematical and Theoretical},
doi = {10.1088/1751-8121/ad17d7}
}

@article{Okounkov_Reshetikhin,
author = {Okounkov, Andrei and Reshetikhin, Nicolai},
year = {2001},
month = {08},
pages = {},
title = {Correlation function of Schur process with application to local geometry of a random 3-dimensional Young diagram},
volume = {16},
journal = {J. Amer. Math. Soc.},
doi = {10.1090/S0894-0347-03-00425-9}
}

@book{Spivak,
	author = {Spivak, Michael},
	title = {Comprehensive introduction to differential geometry. Vol.4, Chapter 9 addendum 1, Vol. 5 Chapter 19, Section 5},
	publisher={Harvard University Press},
        year = {1979},
}

@book{AKH,
    author = {N.I.Akhiezer},
    title = {Elements of the theory of elliptic functions, Chapter 3, Paragraph 14},
    publisher = {AMS vol 79},
    year = {1990}
}

@article{Propp_aztec,
    author = {James, Propp},
    title = {Generalized domino-shuffling},
    journal = {Theoretical Computer Science
303(2):267–301. Tilings of the Plane},
    year = {2003}
}

@article{korn,
    author = {Korn~A.},
    title = {Zwei Anwendungen der Methode der sukzessiven Anndherungen},
    journal = {Schwarz
 Festschrift,pp.215-229},
    year = {1914}
}

@article{Lichtenstein,
    author = {Lichtenstein~L.},
    title = {Zur Theorie der konformen Abbildung. Konforme Abbildung nichtana- lytischer, singularitdtenfreier Fldchenstficke auf ebene Gebiete},
    journal ={Bull. Int. del'Acad.
 Sci. Cracovie, ser. A pp.~192-217},
    year = {1916}
}

@article{Bobenko_dimers,
    author = {Bobenko~A., Bobenko~N., Suris~Yuri.},
    title ={Dimers and M-Curves},
    journal = {preprint, arxiv-2402.08798},
    year = {2024} 
}

@article{Sunil_doubly,
author = {Chhita, Sunil and Johansson, Kurt},
year = {2014},
month = {10},
pages = {},
title = {Domino statistics of the two-periodic Aztec diamond},
volume = {294},
journal = {Advances in Mathematics},
doi = {10.1016/j.aim.2016.02.025}
}

@book{kari,
author = {Astala, Kari and Iwaniec, Tadeusz and Martin, Gaven},
year = {2009},
month = {01},
publisher={Princeton University Press},
title = {Elliptic Partial Differential Equations and Quasiconformal Mappings in the Plane},
journal = {Princeton University Press: Princeton Mathematical Series}
}

@article{Kasteleyn,
    author = {Kasteleyn, P.M.},
    title = {The Statistics of Dimers on a Lattice},
    journal ={Physica, 27, 1209-1225.},
    year ={1961} 
}

@article{BCdT:ellipt,
  author = {Boutillier, C\'edric and Cimasoni, David and de Tilière, Béatrice},
  title = {Elliptic dimer models on minimal graphs and genus 1 Harnack curves},
  journal = {Communications in Mathematical Physics},
  doi = {10.1007/s00220-022-04612-6},
  url = {https://doi.org/10.1007/s00220-022-04612-6},
  year = 2023,
  month = feb,
  volume = 400,
  number = 2,
  pages = {1071--1136},
  archiveprefix = {arXiv},
  eprint = {2007.14699},
  primaryclass = {math.PR},
  hal = {hal-02908609},
  month_numeric = 2,
}

@misc{boutillier2024focks,
      title={Fock's dimer model on the Aztec diamond}, 
      author={Cédric Boutillier and Béatrice de Tilière},
      year={2024},
      eprint={2405.20284},
      archivePrefix={arXiv},
    }

@article{TF,
    author = {Harold N.~V. Temperley and Michael~E. Fisher},
    title = {Dimer problem in statistical mechanics-an exact result},
    journal = {Philosophical magazine},
    year = {6(68):1061--1063, 1961}
}

@article{Borodin-Berggren,
    author ={ Tomas~Berggren and Alexei~Borodin},
    title ={Geometry of doubly-periodic Aztec dimer model},
    journal ={COMMUNICATIONS OF THE
AMERICAN MATHEMATICAL SOCIETY
Volume 5, Pages 475–570 (September 5, 2025)
https://doi.org/10.1090/cams/52},
    year ={2025(preprint version 2023)} 
}

@book{akhiezer,
    author = {Naum Ilich Akhiezer},
    title = {Elements of the theory of elliptic functions},
    publisher = {American mathematical society},
    year = {1990} 
}

\end{document}